\documentclass[hidelinks,11pt,english]{article} 
\usepackage{etex,etoolbox}
\usepackage[pdftex]{graphicx}
\usepackage{multirow}
\usepackage{adjustbox}
\usepackage{setspace}
\usepackage[margin=1in]{geometry}  
\graphicspath{ {./images/} }
\usepackage[dvipsnames]{xcolor}
\usepackage{bm}
\usepackage{pgf, tikz}
\usepackage{amssymb}
\usepackage{rotating}
\usetikzlibrary{positioning}
\usepackage{adjustbox}
\usetikzlibrary{graphs}
\usepackage{caption}
\usetikzlibrary{decorations.pathreplacing}
\usetikzlibrary{arrows, automata}
\usepackage[ruled]{algorithm2e}

\SetCommentSty{mycommfont}
\usepackage{pgfplots}
\pgfplotsset{compat=newest}
\usepackage{hyperref}	
\usepackage{float}
\usepackage{amsfonts}
\usepackage{booktabs}
\usepackage{subfigure}
\usepackage{amsmath}
\usepackage{mathrsfs}
\usepackage{datetime}
\usepackage{enumerate}
\usepackage{booktabs}
\usepackage{amsthm}
\usepackage[round]{natbib}

\usepackage[utf8]{inputenc}
\usepackage{color}

\let\origappendix\appendix 
\renewcommand\appendix{\clearpage\pagenumbering{roman}\origappendix}

\usepackage{verbatim}
\usepackage{sectsty}
\makeatletter
\providecommand{\@fourthoffour}[4]{#4}
\def\fixstatement#1{%
	\AtEndEnvironment{#1}{%
		\xdef\pat@label{\expandafter\expandafter\expandafter
			\@fourthoffour\csname#1\endcsname\space\@currentlabel}}}

\globtoksblk\prooftoks{1000}
\newcounter{proofcount}

\long\def\proofatend#1\endproofatend{%
	\edef\next{\noexpand\begin{proof}[Proof of \pat@label]}%
		\toks\numexpr\prooftoks+\value{proofcount}\relax=\expandafter{\next#1\end{proof}}
	\stepcounter{proofcount}}

\def\printproofs{%
	\count@=\z@
	\loop
	\the\toks\numexpr\prooftoks+\count@\relax
	\ifnum\count@<\value{proofcount}%
	\advance\count@\@ne
	\repeat}
\makeatother

\newtheorem{theorem}{Theorem}
\newtheorem{proposition}{Proposition}
\newtheorem{lemma}{Lemma}

\usepackage{xargs}
\usepackage[colorinlistoftodos,prependcaption,textsize=tiny]{todonotes}
\newcommandx{\unsure}[2][1=]{\todo[linecolor=red,backgroundcolor=red!25,bordercolor=red,#1]{#2}}
\newcommandx{\change}[2][1=]{\todo[linecolor=blue,backgroundcolor=blue!25,bordercolor=blue,#1]{#2}}
\newcommandx{\info}[2][1=]{\todo[linecolor=OliveGreen,backgroundcolor=OliveGreen!25,bordercolor=OliveGreen,#1]{#2}}
\newcommandx{\improvement}[2][1=]{\todo[linecolor=Plum,backgroundcolor=Plum!25,bordercolor=Plum,#1]{#2}}
\newcommandx{\thiswillnotshow}[2][1=]{\todo[disable,#1]{#2}}

\theoremstyle{definition}
\newtheorem{definition}{Definition}
\newtheorem{example}{Example}
\newcommand{\argmax}{\operatornamewithlimits{argmax}}

\fixstatement{theorem}
\fixstatement{lemma}
\fixstatement{proposition}
\fixstatement{corollary}

\begin{document}
	\title{Assortment and Price Optimization  \\Under the Two-Stage Luce model}
	\author{
		Alvaro Flores\footnote{College of Engineering \& Computer Science, Australian National University, Australia.}
		\and
		Gerardo Berbeglia\footnote{Melbourne Business School, The University of Melbourne, Australia.}
		\and
		Pascal Van Hentenryck\footnote{H. Milton Stewart School of Industrial and Systems Engineering
			Georgia Institute of Technology, USA.}
	}
	\date{\today}
	
	\maketitle
	
	\begin{abstract}
This paper studies assortment and pricing optimization problems
under the Two-Stage Luce model (2SLM), a discrete choice model
introduced by \cite{echenique2018} that generalizes the multinomial logit model (MNL). The model employs an utility function as in the
the MNL, and a dominance relation between products. When consumers are offered an
assortment $S$, they first discard all dominated products in
$S$ and then select one of the remaining products using the standard MNL.
This model may violate the regularity condition, which states that the probability of
choosing a product cannot increase if the offer set is enlarged. Therefore, the 2SLM
falls outside the large family of discrete choice models based on 
random utility which contains almost all choice models studied in revenue management.
We prove that the assortment problem under the 2SLM is
polynomial-time solvable. Moreover, we show that the capacitated
assortment optimization problem is NP-hard and but it admits
polynomial-time algorithms for the relevant special cases cases where (1) the dominance
relation is attractiveness-correlated and (2) its transitive
reduction is a forest. The proofs exploit a strong connection
between assortments under the 2SLM and independent sets in
comparability graphs. Finally, we study the associated joint pricing and assortment problem under this model.
 First, we show that well known optimal pricing policy for the MNL can be arbitrarily bad. 
 Our main result in this section is the development of an efficient algorithm for this pricing problem.
  The resulting optimal pricing strategy is simple to describe: it assigns the same price for all products,
   except for the one with the highest attractiveness and as well as for the one with the lowest attractiveness.
	
	\end{abstract}

\section{Introduction}\label{sec:intro}
Revenue Management (RM) is the managerial practice of modifying the availability and the prices of products in order to maximise revenue or profit.
The origin of this discipline dates back to the $1970$'s, following the deregulation of the US airline market. A large volume of research
has been devoted to this area over the last $45$ years, with successful results in many industries
ranging from airlines, hospitality, retailing, and others \citep{mcgillRM1999,KokAssortment,Vulcano2010}.

Two main problems lay in the core of RM theory and practice: the optimal assortment problem,
and the pricing problem. The optimal assortment problem consists of selecting a subset of products to
offer customers in order to maximize revenue. Consider, for example,
a retailer with limited space allocated to mobile phones. If
the store has more than $500$ mobile phones that can be acquired
through its distributors (in various combinations of brands and sizes)
and the mobile phone aisle has capacity to fit $50$ phones on the
shelves, the store manager has to decide which subset of products to
offer given the product costs and the customer preferences.

In order to solve the assortment problem we need a model to predict how
customers select products when they are presented with a set of
alternatives. Most models of discrete choice theory postulate that consumers assign an
\textit{utility} to each alternative and given an offer set, they would choose the alternative with maximum utility. Different assumptions on the distribution of the utilities lead to different discrete choice
models: Celebrated examples include the multinomial logit (MNL)
\citep{luce1959}, the mixed multinomial logit (MMNL)
\citep{daly1978improved}, and the nested multinomial logit (NMNL)
\citep{williams1977formation}.

The multinomial logit model (MNL), also known as the Luce model, is
widely used in discrete choice theory. Since the model was introduced
by \cite{luce1959}, it was applied to a wide variety of demand
estimation problems arising in transportation
\citep{mcfadden1978modeling,catalano2008car}, marketing
\citep{Gudagni,Gensch,rusmevichientong2010dynamic}, and revenue
management \citep{talluri2004revenue,rusmevichientong2010dynamic}. One
of the reasons for its success stems from its small number of
parameters (one for each product): This allows for simple estimation
procedures that generally avoids over fitting problems even when there is limited historical data
\citep{McFadden1974}. However, one of the flaws of the MNL is the property known as the {\em Independence of
	Irrelevant Alternatives}(IIA), which states that the ratio between
the probabilities of choosing elements $x$ and $y$ is constant
regardless of the offered subset. This property does not hold when
products cannibalize each other or are perfect substitutes
\citep{ben1985,Debreu1960,anderson1992}.

Several extensions to the MNL model have been introduced to overcome
the IIA property and some of its other weaknesses; They include the
nested multinomial logit and the latent class MNL model.  These models
however do not handle zero-probability choices well. Consider two
products $a$ and $b$: The MNL model states that the probability of
selecting $a$ over $b$ depends on the relative attractiveness of $a$ compared
to the attractiveness of $b$. Consider the case in which $b$ is never
selected when $a$ is offered. Under the MNL model, this means that $b$
must have zero attractiveness. But this would prevent $b$ from being
selected even when $a$ is not offered in an assortment.

On the other hand, the pricing problem amounts to determine the prices
that a company should offer, in order to best meet its objectives
(profit maximization, revenue maximization, market share maximization,
etc.), while taking into consideration how customers will respond to
different prices and the interaction between price and the intrinsic
features that each product possess.


This paper considers both problems mentioned before,
for the case when customers follow the Two-Stage Luce model (2SLM). The 2SLM
was recently introduced by \cite{echenique2018} and unlike the MNL,
it allows for violations to the IIA property and regularity \citep{Berbeglia2017Rev}.
The Two-Stage Luce model generalizes the MNL by
incorporating a dominance (anti-symmetric and transitive) relation
among the alternatives. Under such relationship, the presence of an
alternative $x$ may prevent another alternative $y$ from being chosen
despite the fact that both are present in the offered assortment. In
this case, alternative $x$ is said to \textit{dominate} alternative
$y$. However, when $x$ is not present, $y$ might be chosen with
positive probability if it is not dominated by any other product
$z$.

An important application of the 2SLM can be found in assortment
problems where there exists a direct way to compare the products over
a set of features. For illustration, consider a telecommunication
company offering phone plans to consumers. A plan is characterized by
a set of features such as price per month, free minutes in peak hours,
free minutes in weekends, free data, price for additional data, and
price per minute to foreign countries. Given two plans $x$ and $y$, we
say that plan $x$ \emph{dominates} plan $y$, if the price per month of
$x$ is less than that of $y$, and $x$ is \emph{at least} as good as
$y$ in every single feature.  In the past, the company offered
consumers a certain set of plans $S_t$ each month $t$ such that no
plan in $S_t$ is dominated by another plan (in $S_t$). The offered
plans however were different each month. Using historical data and
assuming that consumers preferences can be approximated using a
multinomial logit, it is possible to perform a robust estimation
procedure to obtain the parameters of such MNL model. Once the
parameters are obtained, the assortment problem consists in finding
the best assortment of phones plans $S^*$ to maximize the expected
revenue. A natural constraint in this problem consisting in enforcing
that every phone plan offered in $S^*$ cannot be dominated by any
other. Section~\ref{sec:model}	shows that the problem discussed here
can be modelled using the 2SLM and thus solving this problem is reduced to solving
an assortment problem under the 2SLM.

\section{Contributions}\label{sec:contributions}
The first key contribution is to show that the assortment problem can be solved in polynomial time under the 2SLM. The proof is built upon two unrelated results in optimization: the polynomial-time solvability of the
maximum-independent set in a comparability graph \citep{Mohring1985} and a seminal result by \cite{megiddo1979combinatorial} that provides an algorithm to
solve a class of combinatorial optimization problems with rational
objective functions in polynomial time. This is particularly appealing since the 2SLM is one of the very few choice models that goes beyond the random utility model and it allows violations the property known as \emph{regularity}: the probability of choosing an alternative cannot increase if the offer set is enlarged. Since many decades ago, there are well-documented lab experiments where the regularity property is violated \citep{huber1982adding,Tversky_reg,herne1997decoy}.

The second
key contribution is to show that the capacitated assortment problem
under the 2SLM is NP-hard, which contrasts with results on the
MNL. We then propose polynomial algorithms for two
interesting subcases of the capacitated assortment problem: (1) When
the dominance relation is attractiveness-correlated and (2) when the
transitive reduction of the dominance relation can be represented as a
forest. The proofs use a strong connection between assortments under
the 2SLM and independent sets.




The third and final contribution, is an in-depth study of the pricing problem
under the 2SLM. We first note that changes in prices should be
reflected in the dominance relation if the differences between the
resulting attractiveness are large enough. This is formalized
by solving the Joint Assortment and Pricing problem under the
Threshold Luce model, where one product dominates another if the ratio
between their attractiveness is bigger than a fixed threshold.  Under
this setting, we show that this problem can be solved in
polynomial time. The proof relies on the following interesting facts:
(1) An intrinsic utility ordered assortment is optimal; (2) the
optimal prices can be obtained in polynomial time; and (3) it assigns
the same price for all products, except for two of them, the highest
and lowest attractiveness ones. Many of these results are extended to
the following cases (1) capacity constrained problems, where the
number of products that can be offered is restricted and (2) position
bias, where products are assigned to positions, altering their
perceived attractiveness.

The rest of the paper is organized as follows:
Section~\ref{sec:literature} presents a review of the literature
concerning assortment optimization and pricing under variations of the
Multinomial Logit. Section~\ref{sec:model} formalizes the 2SLM and
some of its properties. Section \ref{sec:assortment} proves that
assortment optimization under the 2SLM is polynomial-time
solvable. Section~\ref{sec:capacitated} presents the results on the
capacitated version, particularly the NP-hardness of the capacitated
version of the problem, but also provide polynomial time solutions for
two special cases. Section~\ref{sec:JAP_TLM} present the results for
pricing optimisation under the Threshold Luce
model. Section~\ref{sec:conclusion} concludes the paper and provides
future research directions. All proofs missing from the main text, are
provided in Appendix \ref{App:proofs}.

\section{Literature Review}
\label{sec:literature}

Since the assortment problem and the joint assortment and pricing
problem are a very active research topic, we focus on recent results
closely related with this paper and in particular, results over the
multinomial logit model (MNL) \citep{luce1959,mcfadden1978modeling}
and its variants.

Despite the IIA property, the MNL is widely used. Indeed, for many
applications, the mean utility of a product can be modeled as a linear
combination of its features. If the features capture the mean utility
associated with each product, then the error between the utilities and
their means may be considered as independent noise and the MNL emerges
as a natural candidate for modeling customer choice. In addition, the
MNL parameters can be estimated from customer choice data, even with
limited data \citep{ford_ranking,negahban2012iterative}, because the
associated estimation problem has a concave log likelihood function
\citep{McFadden1974} and it is possible to measure how good the fitted
MNL approximates the data \citep{Hausman1984}. Moreover, it is
possible to improve model estimation when the IIA property is likely
to be satisfied \citep{train2003}.

One of the first positive results on the assortment problem under the
multinomial logit model was obtained by \cite{talluri2004revenue},
where the authors showed that the optimal assortment can be found by
greedily by adding products to the offered assortment in the order of
decreasing revenues, thus evaluating at most a linear number of
subsets. \cite{rusmevichientong2010dynamic} studied the assortment
problem under the MNL but with a capacity constraint limiting the
products that can be offered. Under these conditions, the optimal
solution is not necessarily a revenue-ordered assortment but it can
still be found in polynomial time.

\cite{gallego2011general} proposed a more general attraction model
where the probabilities of choosing a product depend on all the
products (not only the offered subset as in the MNL). This involves a
shadow attraction value associated with each product that influence
the choice probabilities when the product is not offered.
\cite{davis2013assortment} showed that a slight transformation of the MNL
model allows for the solving of the assortment problem when the choice
probabilities follow this more sophisticated attraction model. This
continues to hold when assortments must satisfy a set of totally
unimodular constraints.

The Mixed Multinomial Logit \citep{daly1978improved} is an extension
of the MNL model, where different sets of customers follow different
MNL models. Under this setting, the problem becomes NP-hard
\citep{bront2009column} and it remains NP-hard even for two customer
types \citep{rusmevichientong2014}. A branch-and-cut algorithm was
proposed by \cite{MendezDiaz2014}.  \cite{FeldmanBounding} proposed
methods to obtain good upper bounds on the optimal
revenue. \cite{RobustMNL2012} considered a model where customers
follow a MNL model and the parameters belong to a compact uncertainty
set. The firm wants to hedge against the worst-case scenario and the
problem amounts to finding an optimal assortment under this
uncertainty conditions. Surprisingly, when there is no capacity
constraint, the revenue-ordered strategy is optimal in this setting.
\cite{Jagabathula} proposed a local-search heuristic for the
assortment problem under an arbitrary discrete choice
model. \cite{davis2013assortment} and \cite{Abeliuk2016} proposed
polynomial time algorithms to solve the assortment problem under the
MNL model with capacity constraint and position bias, where position
bias means that customer choices are affected by the positioning of
the products in the assortment. Recently, \cite{jagabathula2015model}
proposed a partial-order model to estimate individual preferences,
where preference over products are modeled using forests. They cluster
the customers in classes, each class being represented with a
forest. When facing an assortment $S$, customers select, following an
MNL model, products that are roots of the forest projected on
$S$. This approach outperformed state-of-the-art methods when
measuring the accuracy of individual predictions.


Attention has also been devoted to discrete choice models to represent
customer choices in more realistic ways, including models that violate
the IIA property \citep{ben1985}. This property does not always hold
in practice \citep{Rieskamp2006}, including when products cannibalize
each other \citep{ben1985}.  \cite{echenique2018palm} identify these violations as
perception priorities, and adjust probabilities to take their effects
into account. \cite{gul2014random} provide an axiomatic generalization
of MNL model to address the case where the products share
features. \cite{fudenberg2015dynamic} propose an axiomatic
generalization of a discounted logit model incorporating a parameter
to model the influence of the assortment size.

Customers tend to use rules to simplify decisions, and before making a
purchase decision, they often narrow down the set of alternatives to
chose from, using different heuristics to make the decision process
simpler. Several models of \textit{consider-then-choose} models have
been proposed in the literature, related with attention filters,
search costs, feature filters, among others, another reasonable way to
discard options, is when the difference between attractiveness is so
evident, that the less attractive alternative, even when it is
offered, is never picked (as in the Threshold Luce model,
\cite{echenique2018}).  Any of the heuristics mentioned before allows
the consumer to restrict her attention to a smaller set usually
referred in the literature as \textit{consideration set}. This effect
also provokes that offered product might result having
zero-probability choices.

Several models have been proposed to address the issue of
zero-probability choices. \cite{masatlioglu2012revealed} propose a
theoretical foundation for maximizing a single preference under
limited attention, i.e., when customers select among the alternatives
that they pay attention to. \cite{manzini2014stochastic} incorporate
the role of attention into stochastic choice, proposing a model in
which customers consider each offered alternative with a probability
and choose the alternative maximizing a preference relation within the
considered alternatives. This was axiomatized and generalized in
\cite{brady2016menu}, by introducing the concept of \textit{random
	conditional choice set rule}, which captures correlations in the
availability of alternatives. This concept also provided a natural way
to model substitutability and complementarity.

\cite{payne1976task} showed that a considerable portion of the subjects in his
experimental setting use a decision process involving a consideration set.
Numerous studies in marketing also validated a consider-then-choose decision process.
In his seminal work \cite{hauser1978testing} observed that most of the
heterogeneity in consumer choice can be explained by consideration sets.
He shows that nearly $80\%$ of the heterogeneity in choice is captured
by a richer model based in the combination of consideration sets and logit-based rankings.
The rationale behind this observation is that first stage filters eliminate a
large fraction of alternatives, thus the resulting consideration sets are
composed of a few products in most of the studied categories
\citep{belonax1978evoked,hauser1990evaluation}. \cite{pras1975comparison} and \cite{gilbride2004}
empirically showed that consumers form their consideration
sets by a \textit{conjunction} of elimination rules.
Furthermore, there are empirical results
showing that a Two-Stage model including consideration sets better fits
consumer search patterns than sequential models \citep{de2012testing}.

Form a customer standpoint, the use for consider-then-choose models alleviate
the cognitive burden of deciding when facing too many alternatives
\cite{Tversky1972choice,Tversky1972elim, tversky1974judgment,payne1996time}.
When dealing with a decision under limited time and knowledge, customers often
recur to screening heuristics as show in \cite{gigerenzer1996reasoning}.
Psychologically speaking, customers as decision makers need to carefully
balance search efforts and opportunity costs with potential gains, and consideration
sets help to achieve that goal \citep{roberts1991consideration,hauser1990evaluation,payne1996time}.
Recently \cite{jagabathula2016nonparametric} proposed a Two-Stage
model where customers consider only the products are contained within
certain range of their willingness to pay. \cite{aouad2015assortment} explored
consider-then-choose models where each costumer has a consideration set, and a ranking
of the products within it. The customer then selects the higher ranked product offered.
The authors studied the assortment problem under several consideration sets and ranking structure,
and provide a dynamic programming approach capable of returning the optimal assortment in polynomial
time for families of consideration set functions originated by screening rules \cite{hauser2009non}.
\cite{dai2014choice} considered a revenue management model where an upcoming customer might discard
one offered itinerary alternative  due to individual restrictions, such as time of departure.
\cite{wang2017impact} studied a choice model that incorporates product search costs,
so the set that a customer considers might differ from what is being offered.

Multi-product price optimisation under the MNL and the NL has been
studied since the models were introduced in the literature. One of
the first results on the structure of the problem is due to
\cite{hanson1996optimizing}, where they show that the profit function
for a company selling substitutable products when customers follow the
MNL model is not jointly concave in price. To overcome this issue, in
\cite{Song2007DemandMA} and later in \cite{dong2009dynamic}, the
authors show that even when the profit function is not concave in
prices, it is concave in the market share and there is a one-to-one
correspondence between price and market share. Multiple studies shown
that under the MNL where all products share the same price sensitivity
parameter, the \textit{mark-up} which is simply the difference between
price and cost, remains constant for all products at optimality
\citep{anderson1992,Hopp2005pricing,gallego2009upgrades,Besbes2016}.
Furthermore, the profit function is also uni-modal on this constant
quantity and it has a unique optimal solution, which can be determined
by studying the first order conditions.

\cite{li2011pricing} showed the same result for the NL model.  Up to
that point, all previous results assumed an identical price
sensitivity parameter for all products.  Under the MNL, there is
empirical evidence that shows the importance of allowing different
price sensitivity parameters for each product
\citep{berry1995automobile,erdem2002impact}. There is is also evidence
in \cite{borsch1990compatibility} that restricting the nest specific
parameters to the unit interval results in rejection of the NL model
when fitting the data, thus recommending to relax this assumption. The
problem when relaxing this condition, is that the profit function is
no longer concave on the market share, which complicates the
optimization task.  In \cite{gallego2014multiproduct} the authors
considered a NL model with differentiated price sensitivities, and
found that the \textit{adjusted mark-up}, defined as price minus cost
minus the reciprocal of the price sensitivity is constant for all
products within a nest at optimality.  Furthermore, each nest also has
an \textit{adjusted next-level markup} which is also invariant across
nests, which reduces the original problem to a one variable
optimization problem. Additional theoretical development can be found
in \cite{rayfield2015approximation,kouvelis2015pbm} but there are
restricted to the Two-Stage nested logit model.  In
\cite{huh2015pricing} some of the results were extended to a
multi-stage nested logit model for specific settings, but also show
that the equal mark-up property fails to hold in general for products
that do not share the same immediate \textit{parent} node in the
nested choice structure, even when considering identical price
sensitivity parameters. \cite{li2011pricing} and
\cite{gallego2014multiproduct} extend to the multi-stage NL model and
show that an optimal pricing solution can still be found by means of
maximizing a scalar function.

There are some interesting results for other models that share
similarities with the MNL, and therefore are closely related with the
model that we are studying.  In \cite{wang2017impact}, the authors
incorporate search cost into consumer choice model.  The results on
this paper for the Joint Assortment and Pricing are similar to the
ones that we study in Section~\ref{sec:JAP_TLM}, in that many
structural results that holds at optimality for their model, are also
satisfied in our studied case. They show that the \textit{quasi-same
	price} policy (that charges the same price for all products but one,
the least attractive one) was optimal for this model. Interestingly,
the Joint Assortment and Pricing results under the Threshold Luce
Model has a slightly different result: The optimal pricing is a fixed
price for all products, except for the most attractive and least
attractive ones. This led to a situation where there are many possible
prices, not just two.

Recently \cite{exponomial2016} hast studied in depth a model which was originally due to
\cite{daganzo1979multinomial} that assumes a negatively skewed
distribution of consumer utilities. The resulting choice
probabilities have an interesting consequence in the optimal pricing
policy: They allow for variable mark-ups in optimal prices that
increase with expected utilities.

The model considered in this paper is a variant of the MNL, proposed
by \cite{echenique2018} and called the \textit{Two-Stage Luce model};
It handles zero-probability choice by introducing the concept of
\textit{dominance}, meaning that if a product $x$ dominates a product
$y$, then $y$ is never selected in presence of $y$. And therefore the
consideration set is formed by considering only non-dominated products
in the offered assortment, allowing flexibility on the consideration
set formation due to the nature of the dominance relation. Once the
consideration set is formed, the customer choose according to an MNL
on the remaining alternatives.  In the following section we describe
this model in detail, and show some examples that highlight many
practical applications for it.
	
\section{The Two-Stage Luce model}
\label{sec:model}

The 2SLM \citep{echenique2018} overcomes a key limitation of the MNL:
The fact that a product must have zero attractiveness if it has zero
probability to be chosen in a particular assortment. This limitation
means that the product cannot be chosen with positive probability in
any other assortment. The 2SLM eliminates this pathological situation
through the concept of \textit{consideration function} which, given a
set of products $S$, returns a subset of $S$ where each product has a
positive probability of being selected.  Let $X$ denotes the set of
all products and let $a(x) > 0$ be the attractiveness of product $x
\in X$. For notational convenience, we use $a_x$ to denote the
attractiveness of product $x$, i.e., $a_x=a(x)$. We extend the
attractiveness function to consider the outside option, with index 0
and $a_0=a(0) \geq 0$, to model the fact that customers may not select
any product. As a result, the attractiveness function has signature
$a:X\cup \left\{0\right\}\to \mathbb{R}^+$.  Given an assortment
$A\subseteq X$, a stochastic choice function $\rho$ returns a
probability distribution over $A$, i.e., $\rho(x,A)$ is the
probability of picking $x$ in the assortment $A$. The 2SLM is a sub
case of the general Luce model presented in \cite{echenique2018}, and
independently discovered in \cite{Ahumada2018}, which is defined
below.

\begin{definition}[General Luce Function \footnote{The definition is slightly different: It makes the outside option effect $a_0$ explicit in the denominator.}, \cite{echenique2018}]\label{def:GLM}
	A stochastic choice function $\rho$ is called a general Luce
	function if there exists an attractiveness function $a\cup
	\left\{0\right\}:X \to \mathbb{R}^+$ and a function $c:2^X\setminus
	\emptyset \to 2^X\setminus \emptyset$ with $c(A)\subseteq A$ for all
	$A\subseteq X$ such that

	\begin{equation}
	\label{2SLM}
	\rho(x,A)=\begin{cases}
	\frac{a_x}{\sum_{y \in c(A)}a_y+ a_0} & \text{if } x\in c(A), \\   
	0 & \text{if } x\notin A.
	\end{cases}
	\end{equation}
	for all $A\subseteq X$. We call the pair $(a,c)$ a general Luce model.
\end{definition}

\noindent
The function $c$ (which is arbitrary) provides a way to capture the
support of the stochastic choice function $\rho$. As observed in
\cite{echenique2018}, there are two interesting cases worthy of being
mentioned:
\begin{enumerate}
	\item If $c(S)$ is a singleton for all $S \subseteq X$, then
	$\rho(x,S)$ is a deterministic choice.
	\item If $c(S)=S$ for all $S\subseteq X$, then the 2SLM coincides with
	the MNL.
\end{enumerate}
\noindent

Two special cases of this model were provided in
\cite{echenique2018}. The first is the two-stage Luce model. This
model restricts $c$, such that the $c(A)$ represents the set of all
\textit{undominated alternatives} in $A$.

\begin{definition}[two-Stage Luce model (2SLM), \cite{echenique2018}]\label{def:2SLM}
	A general Luce model $(a,c)$ is called a 2SLM if there exists a
	strict partial order (i.e. transitive, antisymmetric and irreflexive
	binary relation) $\succ$ such that:
	
	\begin{equation}
	\label{def_c}
	c(A)=\left\{x\in A \ \mid \ \not \exists y\in A:  y\succ x\right\}.
	\end{equation}
	We call $\succ$ \textit{dominance relation}.
\end{definition}

\noindent
As a result, any 2SLM can be described by an irreflexive, transitive,
and antisymmetric relation $\succ$ that fully captures the relation
between products. The second model presented in \cite{echenique2018},
which is a particular case of the 2SLM, is the {\em Threshold Luce
	Model} (TLM), where they explain dominance in terms of how big the
attractiveness are when compared with each other, so $c$ is strongly
tied to $a$.  More specifically, for a given threshold $t>0$, the
consideration set $c(S)$ for a set $S\subseteq X$ is defined as:

\begin{equation}
\label{cS-threshold}
c(S)=\{y \in S \ \mid \  \not\exists x\in S: a_x>(1+t)a_y\}.
\end{equation}

\noindent
In other words, $x \succ y$ if and only if $\frac{a_x}{a_y}>(1+t)$.
Intuitively, an attractiveness ratio of more than $(1+t)$ means that
the less-preferred alternative is dominated by the more-preferred
alternative. Observe that the relation $\succ$ is clearly irreflexive,
transitive, and antisymmetric.

The dominance relation $\succ$ can thus be represented as a Directed
Acyclic Graph (DAG), where nodes represent the products and there is a
directed edge $(x,y)$ if and only if $x \succ y$. Sets satisfying
$c(S)=S$ are \textit{anti-chains} in the DAG, meaning that there are
no arcs connecting them. For instance, consider the {\em Threshold
	Luce model} defined over $X=\{1,2,3,4,5\}$ with attractiveness
values $a_1=12,a_2=8,a_3=6,a_4=3 \text{ and }a_5=2$, and threshold
$t=0.4$. We have that $i \succ j$ iff $a_i > 1.4 \ a_j$.

The DAG representing this dominance relation is depicted in Figure
\ref{fig:Threshold}.

\begin{figure}[!ht]
	\centering
	\begin{tikzpicture}[
	> = stealth, 
	shorten > = 1pt, 
	auto,
	node distance = 3cm, 
	semithick 
	]
	
	\tikzstyle{every state}=[
	draw = black,
	thick,
	fill = white,
	minimum size = 4mm
	]
	
	\node[state] (1) [label=below:{\small $a_1=12$}]{$1$};
	\node[state] (2) [label=below:{\small $a_2=8$}][right  of=1] {$2$};
	\node[state] (3) [label=below:{\small $a_3=6$}][right  of=2] {$3$};
	\node[state] (4) [label=below:{\small $a_4=3$}][right  of=3] {$4$};
	\node[state] (5) [label=below:{\small $a_5=2$}][right  of=4] {$5$};

	
	
	\path[->] (1) edge  (2);
	\path[->] (1) edge[bend left]  (3);
	\path[->] (1) edge[bend left]  (4);
	\path[->] (1) edge[bend left]  (5);
	\path[->] (2) edge[bend left]  (4);
	\path[->] (2) edge[bend left]  (5);	
	\path[->] (3) edge  (4);
	\path[->] (3) edge[bend left]   (5);
	\path[->] (4) edge  (5);
	
	\end{tikzpicture}
	\caption{Example of a DAG for a General Threshold Luce model} \label{fig:Threshold}
\end{figure}
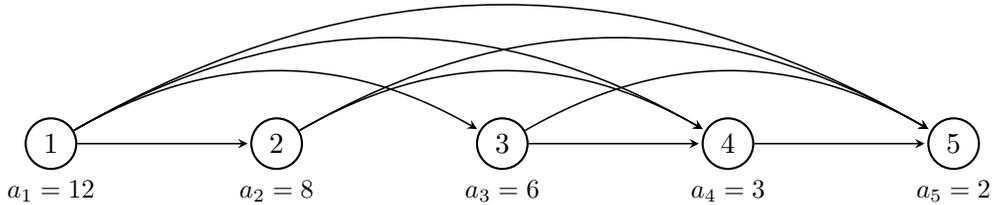

In the following example, we show that the 2SLM admits regularity
violations, meaning that it is possible that the probability of
choosing a product can increase when we enlarge the offered set. Since
regularity is satisfied by any choice model based on random utility
(RUM), this shows that the 2SLM is not contained in the RUM class
\footnote{Observe that this implies that the 2SLM is not contained by
	the Markov chain model proposed by \citep{blanchet2016markov} since
	this last one belongs to the RUM class
	\citep{berbeglia2016discrete}.}.

\begin{example}\label{ex:reg_violation}
	Consider the following instance of the Threshold Luce model (which is a special case of the 2SLM).
	Let $X=\{1,2,3,4\}$ with attractiveness $a_1=5,a_2=4,a_3=3$ and $a_4=3$.
	Consider $t=0.4$ and the attractiveness of the outside option $a_0=1$.
	For the offer set $\{2,3,4\}$, the probability of selecting product $2$
	is $4/11$ since no product dominates each other. However, if we add product $1$
	to the offer set, i.e. if we offer all four products, then the probability of
	selecting product $2$ increases to $4/10$, because products $3$ and $4$
	are now dominated by product $1$.
\end{example}

The Two-Stage Luce Model allows to accommodate different decision
heuristics and market scenarios by specifying the dominance relation
responding to a specific set of rules. Two cases where this can be
observed are provided below.

\noindent
\textbf{Feature Difference Threshold:} Assume that each product has a
set of features $\mathcal{F}=\{1,\ldots,m\}$. A product $x$ can then
be represented by a $m$-dimensional vector $x\in \mathbb{R}^m$. Assume
that the perceived relevance of each feature $k$ is measured by a
weight $\nu_k$, so that the utility perceived by the customers can be
expressed as a weighted combination of their features
$u(x)=\sum_{k=1}^{m}\nu_k\cdot x_k$.  The dominance relation can be
defined as $x\succ y \iff u(x)-u(y)=\sum_{k=1}^{m}\nu_k(x_k-y_k)\geq
T$, where $T>0$ is a tolerance parameter that represents how much
difference a customer allows before considering that an alternative
dominates another. The dominance relation is irreflexive, transitive,
and antisymmetric and hence it can be used to define an instance of
the 2SLM. One can easily show that this model is a special case of the
TLM.

\noindent
\textbf{Price levels:} Suppose we have $N$ products, each product $i$
has $k_i$ price levels.  Let $x_{il}$ be product $i$ with price
$p_{il}$ attached and it corresponding attractiveness $a_{il}$, we
assume that for each product $i$ prices $p_{ik}$ satisfy $p_{i1}<
p_{i2}<\ldots,p_{ik_i}$.  Naturally, $x_{i1}\succ
x_{i2}\succ\ldots\succ x_{ik_i}$, because for the same product the
customer is going to select the one with the lowest price available.
Each price level for each product can still dominate or be dominated
by other products as well, as long as the dominance relation is
irreflexive, transitive and antisymmetric.  This setting can be
modelled by the \textit{Two-Stage Luce model} in a natural way.

\section{Assortment Problems Under the Two-Stage Luce model}
\label{sec:assortment}

This section studies the assortment problem for the 2SLM using the
definitions and notations presented earlier. Let $r:X\cup
\left\{0\right\} \rightarrow \mathbb{R^+}$ be a revenue function
associated with each product and satisfying $r(0)=0$. The expected
revenue of a set $S\subseteq X$ is given by
\begin{equation}
\label{ER}
R(S)=\sum_{i\in c(S)}\rho(i,S)r(i).
\end{equation}

\noindent
The assortment problem amounts to finding a set
\[
S^*\in\argmax_{S\subseteq X}R(S)
\]
yielding an optimal revenue of
\[
R^*=\max_{S\subseteq X}R(S).
\]
Observe that every subset $S \subseteq X$ can be uniquely represented by a
binary vector $x \in \left\{0,1\right\}^n$ such that $i \in S$ if and
only if $x_i=1$. Using this bijection, the search space for $S^*$ can
be restricted to
\[
\mathcal{D}=\{x\in\{0,1\}^n \mid \ \forall s \succ t: x_s+x_t\leq 1 \}
\]
where $\mathcal{D}$ represents all the subsets satisfying $S=c(S)$,
which means that no product on $S$ dominates another product in
$S$. There is always an optimal solution $S^*$ that belongs to
$\mathcal{D}$ because $R(S)=R(c(S))$ and $c(S) \in D$ for all sets $S$
in $X$. As a result, the Assortment Problem under the 2SLM (\texttt{AP-2SLM}) can be
formulated as
\begin{equation*}
\label{Assortment_General}
\tag{\texttt{AP-2SLM}}
\begin{aligned}
& \underset{x}{\text{maximize}}
& & \frac{\sum_{i=1}^{n}r_ia_ix_i}{\sum_{i=1}^{n}a_ix_i +a_0} \\
& \text{subject to}
& & x\in \mathcal{D}
\end{aligned}
\end{equation*}
where $r_i$ and $a_i$ represent $r(i)$ and $a(i)$ for simplicity.

An effective strategy for solving many assortment problems consists in
considering revenue-ordered assortments, which are obtained by
choosing a threshold $\rho$ and selecting all the products with
revenue at least $\rho$. This strategy leads to an optimal algorithm
for the assortment problem under the MNL. Unfortunately, it fails
under the 2SLM because adding a highly attractive product may remove many
dominated products whose revenues and utilities would lead to a higher revenue.

\begin{example}[Sub-Optimality of Revenue-Ordered Assortments]
	\label{rev_ord_fail}
	Consider a Threshold Luce model with $X=\{1,2,3\}$, revenues
	$r_1=88,r_2=47,r_3=46$, attractiveness $a_0=55,a_1=13,a_2=26,a_3=15$
	and $t=0.6$. Then $x\succ y$ iff $a_x> 1.6 \ a_y$ which gives $2
	\succ 1$ and $2 \succ 3$.  Consider the sets $S \subseteq X$
	satisfying $S=c(S)$:
	
	\begin{table}[H]
		\begin{center}
			\begin{tabular}{@{} cc @{}}
				\toprule
				$S$ &        $R(S)$  \\
				\midrule
				$\{1\}$  & 16.824    \\	
				$\{2\}$ & 15.086   \\
				$\{3\}$  & 9.857   \\
				$\{1,3\}$ & 22.096    \\
				\bottomrule
			\end{tabular}
		\end{center}
	\end{table}
	
	The optimal revenue is given by assortment $\{1,3\}$, while the best
	revenue-ordered assortment under the 2SLM is $S = \{1\}$, yielding almost $24\%$
	less revenue.
\end{example}

\noindent


%

To solve problem \ref{Assortment_General}, consider first the
\texttt{MaxAtt} problem defined over the same set of
constraints. Given weights $c_i \in \mathbb{R}$ $(1 \leq i \leq n)$,
the \texttt{MaxAtt} problem is defined as follows:

\begin{equation*}
\label{Assortment_Linear}
\tag{\texttt{MaxAtt}}
\begin{aligned}
& \underset{x}{\text{maximize}}
& & \sum_{i=1}^{n}c_ix_i \\
& \text{subject to}
& & x\in \mathcal{D}
\end{aligned}
\end{equation*}

\noindent
We now show that \eqref{Assortment_Linear} can be reduced to the
maximum weighted independent set problem in a directed acyclic graph
with positive vertex weights. An independent set is a set of vertices
$I$ such that there is no edge connecting any two vertices in $I$. The
maximum weighted independent set problem (\texttt{MWIS}) can be stated
as follows:

\begin{definition}
	\label{MWIS}
	\textit{Maximum Weighted Independent Set Problem:} Given a graph
	$G=(V,E)$ with a weight function $w:V\rightarrow\mathbb{R}$, find an
	independent set $I^*\in \argmax_{I\in \mathcal{I}}\sum_{i\in
		I}w(i)$, where $\mathcal{I}$ is the set of all independent sets.
\end{definition}

\noindent
Recall that the dominance relation can be represented as a DAG $G$
which includes an arc $(u,v)$ whenever $u \succ v$. As a result, the
condition $x\in \mathcal{D}$ implies that any feasible solution to
\eqref{Assortment_Linear} represents an independent set in $G$ and
maximizing $\sum_{i=1}^{n}c_ix_i$ amounts to finding the independent
set maximizing the sum of the weights. Since the dominance relation is
a partial order, the DAG representing the dominance relation is a
comparability graph. The following result is particularly
useful.

\begin{theorem}[\cite{Mohring1985}]
	\label{comparability}
	The maximum weighted independent set is polynomially-solvable for
	comparability graphs with positive weights.
\end{theorem}

\noindent
We are ready to present our first result.

\begin{lemma}
	\label{eq_A_MWIS}
	\eqref{Assortment_Linear} is polynomial-time solvable.
\end{lemma}

\begin{proof}
We first show that we can ignore those products with a negative
weight. Let $\hat{X}=\{i \in X\mid c_i>0\}$ and
$\mathcal{\hat{D}}=\{x\in\{0,1\}^n \mid \forall s, t \in \hat{X}, s \succ t: \ x_s+x_t\leq 1\}$.
Solving \eqref{Assortment_Linear} is equivalent to solving:

\begin{equation*}
\tag{\texttt{Reduced MaxAtt}}
\begin{aligned}
& \underset{x}{\text{maximize}}
& & \sum_{i\in\hat{X} }c_ix_i \\
& \text{subject to}
& & x\in \mathcal{\hat{D}}
\end{aligned}
\label{Assortment_reduced}
\end{equation*}

\noindent
Indeed, consider an optimal solution $x^*$ to Problem
\ref{Assortment_Linear} and assume that there exists $i\in X$ such
that $c_i<0$ and $x_i^*=1$. Define $\hat{x}$ like $x^*$ but with
$\hat{x_i}=0$. $\hat{x}$ has a strictly greater value for the objective function
in \ref{Assortment_reduced} than $x^*$ has, and
is feasible since setting a component to zero cannot violate any
constraint (i.e., $\hat{x} \in \mathcal{D}$). This contradicts the
optimality of $x^*$. Now Problem \texttt{Reduced MaxAtt} can be
reduced to solving an instance of Problem \texttt{MWIS} in a DAG with
positive weights that corresponds to the dominance relation. This DAG
is a comparability graph and the result follows from Theorem
\ref{comparability}.
\end{proof}	

The next step in solving the assortment problem under the 2SLM relies
on a result by Megiddo \cite{megiddo1979combinatorial}. Let $D$ be a
domain defined by some set of constraints and consider Problem
\ref{P_A}
\begin{equation*}
\tag{\texttt{A}}
\begin{aligned}
& \underset{x}{\text{maximize}}
& & \sum_{i=1}^{n}c_ix_i \\
& \text{subject to}
& & x\in D
\end{aligned}
\label{P_A}
\end{equation*}

\noindent
and its associated Problem \ref{P_B}:

\begin{equation*}
\tag{\texttt{B}}
\begin{aligned}
& \underset{x}{\text{maximize}}
& & \frac{a_0+\sum_{i=1}^{n}a_ix_i}{b_0+\sum_{i=1}^{n}b_ix_i} \\
& \text{subject to}
& & x\in D.
\end{aligned}
\label{P_B}
\end{equation*}

\noindent
Using this notation, Megiddo's theorem can be stated as follows.

\begin{theorem}[\cite{megiddo1979combinatorial}]
	\label{rational}
	If Problem \ref{P_A} is solvable within $O(p(n))$ comparisons and
	$O(q(n))$ additions, then Problem \ref{P_B} is solvable in
	$O(p(n)(q(n) +p(n)))$ time.
\end{theorem}

\noindent
We are now in position to state our main theorem of this section.

\begin{theorem}\label{thm:main}
	The assortment problem under the Two-Stage Luce model is polynomial-time solvable.
\end{theorem}
\begin{proof}
Recall that the assortment problem under the 2SLM (\texttt{AP-2SLM}) can be formulated as
\begin{equation}
\label{Assortment_General_formula}
\begin{aligned}
& \underset{x}{\text{maximize}}
& & \frac{\sum_{i=1}^{n}r_ia_ix_i}{\sum_{i=1}^{n}a_ix_i +a_0} \\
& \text{subject to}
& & x\in \mathcal{D}
\end{aligned}
\end{equation}
where $\mathcal{D}=\{x\in\{0,1\}^n \mid \ \forall s \succ t: x_s+x_t\leq 1 \}.$

The problem of maximizing the numerator in
\eqref{Assortment_General_formula} is exactly the \texttt{MaxAtt}
problem. By Lemma \ref{eq_A_MWIS}, this is polynomial-time
solvable. Now observe that \eqref{Assortment_General_formula} (i.e.,
problem \ref{Assortment_General}) can be seen as a Problem
\ref{P_B}. Therefore, by Theorem \ref{rational}, the assortment
problem under the 2SLM is solvable in polynomial time. 
\end{proof}	
\noindent
In addition to solving the assortment problem under the 2SLM, Theorem
\ref{thm:main} is interesting in that it solves the assortment problem
under a Multinomial Logit with a specific class of constraints.  It
can be contrasted with the results by \cite{davis2013assortment},
where feasible assortments satisfy a set of totally unimodular
constraints. They show that the resulting problem can be solved as a
linear program. However, the 2SLM introduces constraints that are not
necessarily totally unimodular as we now show.

\begin{example}
	\label{2SLMvsTUM}
	Consider $X=\left\{1,2,3,4\right\}$ and $1 \succ 3, 1 \succ 4, 2 \succ
	3, 2 \succ 4,$ and $3 \succ 4$.  The constraint matrix that defines the
	feasible space ($\mathcal{D}$) for this instance is:
	\begin{equation*}
	M=	\begin{bmatrix}
	1       & 0 & 1 & 0 \\
	1       & 0 & 0 & 1 \\
	0       & 1 & 1 & 0 \\
	0       & 1 & 0 & 1 \\
	0       & 0 & 1 & 1 \\
	\end{bmatrix}
	\end{equation*}
	where each row represents a constraint $x_u+x_v\leq 1$.  meaning that
	just one end of the edge can be selected at the time.
	\cite{camion1965characterization} proved that $M$ is totally
	unimodular if and only if, for every (square) Eulerian submatrix A of
	$M$, $\sum_{i,j} a_{ij} \equiv 0 \pmod 4$. Consider the sub-matrix
	corresponding to the first, second, and fifth rows and the first,
	third, and fourth columns
	\begin{equation*}
	N=	\begin{bmatrix}
	1       &1 & 0 \\
	1       & 0  & 1 \\
	0       & 1 & 1  \\
	
	\end{bmatrix}
	\end{equation*}
	Matrix $N$ is eulerian (The sums of every element on each row or on
	each column is a multiple of 2). But the sum of all elements of $N$ is
	$6 \not \equiv 0 \pmod 4$ and hence $M$ is not totally unimodular.
\end{example}

We close this section by explaining how our results can be extended to
a more general setting.  \cite{gallego2014general} proposed the
general attraction model (GAM) to describe customer behaviour, that
alleviates some deficiencies of the MNL.  More specifically, the
intuition behind this choice model is that whenever a product is not
offered, then its absence can potentially increase the probability of
the no-purchase alternative, as consumers can potentially look for the
product elsewhere, or at a later time. To achieve this effect, for
each product $j$ the model considers two different weights: $v_j$ and
$w_j$, usually with $0\leq w_j \leq v_j$.  If product $j$ is offered,
then its preference weight is $v_j$. But if $j$ is not offered, then
the preference weight of the outside option is increased by $w_j$.
For all $j\in X$, let $\widetilde{v_j}=v_j-w_j$ and $\tilde{v_0}=v_0 +
\sum_{k \in X} w_k$. Using this notation, the probabilities associated
with the GAM model can be recovered by means of the following
equation:

\begin{equation}\label{eq:gam}
\rho(j,S)=\begin{cases}
\frac{v_j}{\sum_{i \in S}\widetilde{v_j} + \widetilde{v_0}} & \text{if } j\in S, \\
0 & \text{if } j\notin S.
\end{cases}
\end{equation}

Observe that the resulting assortment problem will has the same
functional form than problem \ref{Assortment_General}, with a slight
modification on the coefficients in the denominator. Thus, we can
apply the same solution technique described in Theorem \ref{thm:main}
to find the optimal assortment for the GAM.

\section{The Capacitated Assortment Problem}
\label{sec:capacitated}

In many applications, the number of products in an assortment is
limited, giving rise to capacitated assortment problems. Let $C$ $(1
\leq C \leq n)$ be the maximum number of products allowed in an
assortment.  The Capacitated Assortment Problem under the Two-Stage Luce Model
(\texttt{C2SLMAP}) is given by
\begin{equation*}
\label{C2SLMAP}
\tag{\texttt{C2SLMAP}}
\begin{aligned}
& \underset{x}{\text{maximize}}
& & \frac{\sum_{i=1}^{n}r_ia_ix_i}{\sum_{i=1}^{n}a_ix_i +a_0} \\
& \text{subject to}
& & x\in \mathcal{D}_C
\end{aligned}
\end{equation*}

\noindent
where $\mathcal{D}_C=\{x\in\{0,1\}^n \mid \forall (s,t)\in \mathcal{R}
\quad x_s+x_t\leq 1 \land \sum_{i=1}^{n}x_i\leq C\}$. As before, it is useful
to define its capacitated maximum-attractiveness counterpart (\texttt{C-MaxAtt}), i.e.,
\begin{equation*}
\label{MACC}
\tag{\texttt{C-MaxAtt}}
\begin{aligned}
& \underset{x}{\text{maximize}}
& & \sum_{i=1}^{n}c_ix_i\\
& \text{subject to}
& & x\in \mathcal{D}_C
\end{aligned}
\end{equation*}

\noindent
This section first proves that the capacitated assortment problem
under the 2SLM is NP-hard. The reduction uses the Maximum Weighted
Budgeted Independent Set (\texttt{MWBIS}) problem proposed by
\cite{DBLP:journals/corr/Bandyapadhyay14} which amounts to finding a
maximum weighted independent set of size not greater than $C$.
\cite{Kalra2017} showed that Problem (\texttt{MWBIS}) is NP-hard for
bipartite graphs.

\begin{theorem}\label{np-hard-proof}
	Problem \eqref{C2SLMAP} is NP-hard (under Turing reductions).
\end{theorem}

\noindent
It is interesting to mention that Problem \eqref{MACC} is equivalent
to finding an anti-chain of maximum weight among those of cardinality
at most $C$. This problem (\texttt{MWLA}) was proposed by
\cite{SHUM1996421} and its complexity was left open, but the above
results show that it is also NP-hard.
\cite{DBLP:journals/corr/Bandyapadhyay14} studied Problem
(\texttt{MWBIS}) for various types of graphs (e.g., trees and
forests), but the dominance relation of the 2SLM can never be a tree
since it is transitive (unless we consider a graph with a single vertex).

In light of this NP-hardness result, the rest of this section presents
polynomial-time algorithms for two special cases of the dominance
relation.

\subsection{The Two-Stage Luce model over Tree-Induced Dominance Relations}
\label{sub:forest}

Let $\mathcal{R}_{\succ}$ be the transitive reduction of the
irreflexive, antisymmetric, and transitive relation $\succ$. This
section considers the capacitated assortment problem when the relation
$\mathcal{R}_{\succ}$ can be represented as a tree. Without loss of
generality, we can assume that the tree contains all
products. Otherwise, we can add another product with zero weight that
dominates all original products. This new product will be the root of
the tree and the products not in the original tree will be the
children of the root. Similarly, the same transformation applies to
the case when $\mathcal{R}_{\succ}$ is a forest. Here all the trees in
the forest will be children of the new product.

We show how to solve Problem (\ref{MACC}). The result follows again by
applying Megiddo's theorem. The first step of the algorithm simply
removes all products with negative weight: Their children can be added
to the parent of the deleted vertex. The main step then solves
(\ref{MACC}) bottom-up using dynamic programming from the leaves. For
simplicity, we present the recurrence relations to compute the weight
of the optimal assortment. It is easy to recover the optimal
assortment itself. The recurrence relations compute two functions:
\begin{enumerate}
	\item $\mathcal{A}(k,c)$ which returns the weight of an optimal
	assortment using product $k$ and its descendants in the tree
	representation of $\mathcal{R}_{\succ}$ for a capacity $c$;
	
	\item $\mathcal{A}^+(S,c)$ which, given a set $S$ of vertices that are
	children of a vertex $k$, returns the weight of an optimal assortment
	using the products in $S$ and their descendants for a capacity $c$.
\end{enumerate}

\noindent
The key intuition behind the recurrence is as follows. If $v$ is a
vertex and $v_1$ and $v_2$ are two of its children, $v_1$ does not
dominate $v_2$ or any of its descendants. Hence, it suffices to
compute the best assortments producing $\mathcal{A}(v_1,0), \ldots,
\mathcal{A}(v_1,C)$ and $\mathcal{A}(v_2,0), \ldots,
\mathcal{A}(v_2,C)$ and to combine them optimally. The recurrence relations
are defined as follows ($v \in X$ and $1 \leq c \leq C$):
\begin{align*}
& \mathcal{A}(v,0)  = 0; \\
& \mathcal{A}(v,c)  = \max(c_v,\mathcal{A}^+(\textit{children}(v),c));
\end{align*}
and
\begin{align*}
& \mathcal{A}^+(\emptyset,0)  = 0; \\
& \mathcal{A}^+(S,c)  = \mathop{\max_{n_1, n_2 \geq 0}}_{n_1 + n_2 = c} \mathcal{A}^+(S \setminus \{e\},n_1) + \mathcal{A}(e,n_2) \mbox{ where } e = \argmax_{i \in S} c_i.
\end{align*}
where {\it children(p)} denotes the children of product $p$ in the
tree. Note that $\mathcal{A}^+(S,c)$ is computed recursively to obtain
the best assortment from the products in $S$ and their descendants. Using these
recurrence relation, the following Theorem follows:
\begin{theorem}
	\label{thm:forest}
	Let $\succ$ a dominance relation whose relation $\mathcal{R}_{\succ}$
	is a tree containing all products. The capacitated assortment problem under the
	2SLM and $\succ$ is polynomial-time solvable.
\end{theorem}

\subsection{The Attractiveness-Correlated Two-Stage Luce model}
\label{sub:UC}

The second special case considers a dominance relation that is correlated with attractiveness.

\begin{definition}[Attractiveness-Correlated Two-Stage Luce model]
	\label{UC2SLM}	
	A Two-Stage Luce model is attractiveness-correlated if the dominance relation
	satisfies the following two conditions:
	\begin{enumerate}
		\item If $x \succ y$, then $a_x > a_y$.
		\item If $x \succ y$ and $a_z > a_x$, then $z \succ y$.
	\end{enumerate}
\end{definition}

\noindent
The first condition simply expresses that product $x$ can only
dominate product $y$ if the attractiveness of $x$ is greater than the
attractiveness of $y$. The second condition ensures that, if $x$
dominates $y$, then any product whose attractiveness is greater than
$x$ also dominates $y$. The induced dominance relation is irreflexive,
anti-symmetric, and transitive. A particular case of this model, is
the Threshold Luce model.

When customers follow the Threshold Luce model, they form their
consideration sets based on the attractiveness of products.  Without
loss of generality, we can assume $a_1\geq a_2\geq \ldots\geq a_n$,
unless stated otherwise. For a set $S$, the associated consideration
set $c(S)$ may be a proper subset of $S$, but for the purpose of
assortment optimization, we don't have incentives to offer sets
including products that are not even consider by customers, so we can
restrict our search for optimal solutions to sets where $c(S)=S$.  A
necessary and sufficient condition for this to happen is
$\frac{\max_{i\in S}a_i}{\min_{i\in S}a_i}\leq 1+t$.  Meaning that
largest ratio between attractiveness is not greater than $1+t$, so no
dominance relation appears.

The firm now needs to carefully balance the inclusion of
high-attractiveness products and their prices to maximize the revenue.
In the following example we show that \textit{revenue ordered
	assortments} are not optimal under the Threshold Luce Model.  In
fact, this strategy can be arbitrarily bad.

\begin{example}[Revenue ordered assortments are not
	optimal] \label{ex:ro_not_optimal} Consider the following product
	configuration. Let $N+1$ products, with prices $p_1$ for the first
	product, and $\alpha p_1$ for the rest of them, with $\alpha<1$. The
	attractiveness for all products is $a_1$ for the first product and
	$\gamma a_1$ for all the rest, such as in the presence of product
	$1$, all the rest of the products are ignored. To complete the set
	up, let $a_0 $ the attractiveness of the outside option.  The best
	revenue ordered assortment is to consider product $1$, given a
	revenue of:
	
	\[R'=R(\{1\})=\frac{p_1 a_1}{a_1+a_0}\]
	
	But, if $N$ is big enough (at least bigger than $\frac{1}{\alpha\gamma}$),
	is more profitable to show $S_N=X\setminus\{1\}$, resulting in a revenue of:
	
	\[R^*=R(S_N)=\frac{N\cdot \alpha p_1\gamma a_1}{N\cdot \alpha \gamma a_1+a_0}\]
	
	Now, if we calculate the ratio if this two values, $R'$ and $R^*$ and let $N$ tend to infinity we have:
	
	\begin{align}
	\frac{R'}{R^*}=&\lim\limits_{N\to\infty}\frac{\frac{p_1 a_1}{a_1+a_0}}{\frac{N\cdot \alpha p_1\gamma a_1}{N\cdot \alpha \gamma a_1+a_0}}\nonumber\\
	\frac{R'}{R^*}=&\lim\limits_{N\to\infty}\frac{p_1 a_1}{a_1+a_0}\cdot\frac{N\cdot \alpha \gamma a_1+a_0}{N\cdot \alpha p_1\gamma a_1}\nonumber\\
	\frac{R'}{R^*}=&\frac{a_1}{a_1+a_0}
	\end{align}
	
	Observe that this last expression is the market share of offering just product $1$,
	which can be arbitrarily bad by either making $a_1$ as small as desired, or making the outside option more attractive.
\end{example}

\noindent
The capacitated assortment optimization can be solved in polynomial time
under the Attractiveness-Correlated Two-Stage Luce model. Consider
an assortment whose product with the largest attractiveness is $k$. This
assortment cannot contain any product dominated by $k$. Moreover, if $k_1$
and $k_2$ are two other products in this assortment, then $k_1$
cannot dominate $k_2$ since $k$ would also dominate $k_2$. As
a result, consider the set
\[
X_k=\{i\in X \mid a_i \leq a_k \ \& \ k \not\succ i \}.
\]
No product in $X_k$ dominates any other product in $X_k$ and hence the
\texttt{C2SLMAP} reduces to a traditional assortment problem under the
MNL. This idea is formalized in Algorithm \ref{MUACC_UC2SLM}, where
\texttt{CMLMAP} is a traditional algorithm for the MNL.  The algorithm
considers each product in turn and the products that it does not
dominate and applies a traditional capacitated assortment optimization
under the MNL. The best such assortment is the solution to the
capacitated assortment under the attractiveness-correlated 2SLM.

\begin{figure}[H]
	\begin{algorithm}[H]\caption{Capacitated Assortment Optimization under the Attractiveness-Correlated 2SLM.}
		\label{MUACC_UC2SLM}
		\DontPrintSemicolon
		\KwData{$X,\succ,r,a$}
		\KwResult{Optimal Assortment $S^*$}
		$R(S^*)=0$
		\For{$k =1,\ldots,n$}{
			$X_k=\{i\in X \mid a_i \leq a_k \ \& \ k \not\succ i\}$ \\
			$S_k=\texttt{CMLMAP}(X_k,r,a)$ \\
			\If{$R(S_k)>R(S^*)$}{
				$S^*=S_k$
			}
		}
		return $S^*$
	\end{algorithm}
\end{figure}

\begin{theorem}
	\texttt{C2SLMAP} can be solved in polynomial time for Attractiveness-Correlated instances.
	\label{thm:AC}
\end{theorem}
\begin{proof}
To show correctness, it
suffices to show that the optimal assortment must be a subset of one
of the $X_k$ $(1 \leq k \leq n)$. Let $A$ be the optimal assortment
and assume that $k$ is its product with the largest attractiveness
(break ties randomly). $A$ must be included in $X_k$ since otherwise
it would contain a product $x$ such that $k \succ x$ (contradicting
feasibility) or such that $a(x) > a(k)$ (contradicting our
hypothesis). The correctness then follows since there is no dominance
relationship between any two elements in each of $X_k$.  The claim of
polynomial-time solvability follows from the availability of
polynomial-time algorithms for the assortment problem under the MNL
and the fact that are exactly $n$ calls to such an algorithm.
\end{proof}

\section{Joint Assortment and Pricing under the Threshold Luce model}
\label{sec:JAP_TLM}

The previous sections provides solutions to the Assortment
Optimization problem under the Two-Stage Luce model.  This section
aims at determining how to assign prices to products in order to
maximise the expected revenue.  It studies the Joint Assortment and
Pricing Problem under the Threshold Luce model, by making the
attractiveness of each product dependent upon its price.  Let
$p=\left(p_1,\ldots,p_n\right)$ be the price vector, where such that
$p_i\in \mathbb{R}_+\cup \{\infty\}$ represents the price of product
$i$.  Since the price will affect the attractiveness $a_i$ of product
$i$, the presentation makes this dependency explicit by writing
$a_i(p_i)$ whose form in this paper is specified by
\begin{equation}
\label{ec:att_price}
a_i(p_i)=\exp(u_i-p_i)
\end{equation}
\noindent
where $u_i$ is the \textit{intrinsic utility} of product $i$ and the
value $v_i=u_i- p_i$ is called the \textit{net utility} of product
$i$.  Assigning an infinite price to a product is equivalent to not
offering the product, as the attractiveness, and therefore the
probability of selecting the product, becomes $0$. Without loss of
generality, products are indexed in a decreasing order by intrinsic
utility.

The following definition is an extension of the definition of a
consideration set given an assortment $S$ when each product $i$ has a
price $p_i$.

\begin{definition}\label{def:consideration_set_prices}
	Given an assortment $S$, a price vector
	$p=\left(p_1,p_2,\ldots,p_n\right)$ and a threshold $t$, the
	consideration set $c(S,p)$ for the Threshold Luce model is defined
	as:
	\begin{equation}
	\label{cS-threshold-p}
	c(S,p)=\{j \in S \ \mid \  \not\exists i\in S: a_i(p_i)>(1+t)a_j(p_j)\}.
	\end{equation}
\end{definition}

The influence of the price vector over the dominance relations is
given by the following example:

\begin{example}\label{ex:price_effect_dominance}[Price effect on the dominance relation]
	Consider the {\em Threshold Luce model} defined over $X=\{1,2,3,4\}$
	with utilities $u_1=\ln(10),u_2=\ln(8),u_3=\ln(6)$ and $u_4=\ln(3)$,
	and consider first a scenario where all products have the same price
	$p_i=\ln(3) \quad \forall i=1,\ldots,4$. Consider also a second
	scenario with prices equal to $p_1'=\ln(4), p_2'=\ln(4),p_3'=\ln(3)$
	and $p_4'=\ln(2)$. For a threshold $t=0.5$, we have that $i \succ j$
	iff $a_i(p_i) > 1.5 \ a_j(p_j)$. A table summarizing the utilities,
	prices, and attractiveness for both scenarios is given in Table
	\ref{table:dagprice} and the DAGs depicting the dominance relations
	for the two scenarios are given in Figures \ref{fig:dag1} and
	\ref{fig:Threshold2}.
	
	\begin{table}[!t]
		\setlength\tabcolsep{2pt} 
		\begin{center}
			\begin{tabular}{@{} cc @{\hspace{1cm}} cc @{\hspace{1cm}} cc @{\hspace{1cm}} cc @{\hspace{1cm}} cc}
				\toprule
				$i$ &        $u_i$ & $p_i$ & $a_i(p_i)$ & $p_i'$ & $a_i(p_i')$   \\
				\midrule
				$1$  & $\ln(10)$ & $\ln(3)$ & $3.\overline{3}$ &$\ln(4)$ &$2.5$   \\	
				$2$ & $\ln(8)$  & $\ln(3)$ & $2.\overline{6}$ &$\ln(4)$ &$2$ \\
				$3$  & $\ln(6)$   & $\ln(3)$ & $2$ &$\ln(3)$ &$2$\\
				$4$ & $\ln(3)$    & $\ln(3)$ & $1$ &$\ln(2)$ &$1.5$\\
				\bottomrule
			\end{tabular}
		\end{center}
		\caption{Summary of utilities, prices and attractiveness for the two proposed scenarios.}
		\label{table:dagprice}
	\end{table}
	
	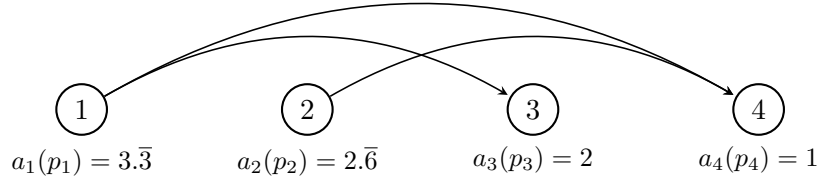
\begin{figure}[!t]
		\centering
		\begin{tikzpicture}[
		> = stealth, 
		shorten > = 1pt, 
		auto,
		node distance = 3cm, 
		semithick 
		]
		
		\tikzstyle{every state}=[
		draw = black,
		thick,
		fill = white,
		minimum size = 4mm
		]
		
		\node[state] (1) [label=below:{\small $a_1(p_1)=3.\overline{3}$}]{$1$};
		\node[state] (2) [label=below:{\small $a_2(p_2)=2.\overline{6}$}][right  of=1] {$2$};
		\node[state] (3) [label=below:{\small $a_3(p_3)=2$}][right  of=2] {$3$};
		\node[state] (4) [label=below:{\small $a_4(p_4)=1$}][right  of=3] {$4$};
		
		
		\path[->] (1) edge[bend left]  (3);
		\path[->] (1) edge[bend left]  (4);
		
		\path[->] (2) edge[bend left]  (4);
		
		

		\end{tikzpicture}
		\caption{The DAG for the first scenario where all prices are fixed to $\ln(3)$ and the threshold is $t=0.5$.
			Product $1$ dominates products $3$ and $4$, and product $2$ dominates product $4$.}
		\label{fig:dag1}
	\end{figure}

	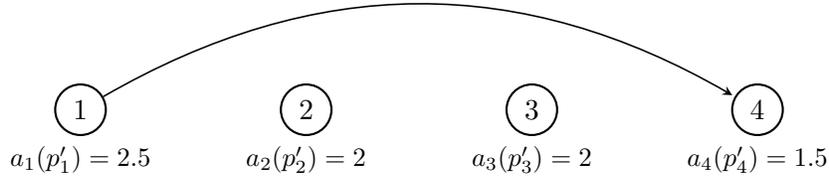
\begin{figure}[!t]
		\centering
		\begin{tikzpicture}[
		> = stealth, 
		shorten > = 1pt, 
		auto,
		node distance = 3cm, 
		semithick 
		]
		
		\tikzstyle{every state}=[
		draw = black,
		thick,
		fill = white,
		minimum size = 4mm
		]
		
		\node[state] (1) [label=below:{\small $a_1(p_1')=2.5$}]{$1$};
		\node[state] (2) [label=below:{\small $a_2(p_2')=2$}][right  of=1] {$2$};
		\node[state] (3) [label=below:{\small $a_3(p_3')=2$}][right  of=2] {$3$};
		\node[state] (4) [label=below:{\small $a_4(p_4')=1.5$}][right  of=3] {$4$};
		

		\path[->] (1) edge[bend left]  (4);
		
		
		

		\end{tikzpicture}
		\caption{The DAG for the second scenario where all prices are fixed to $(\ln(4),\ln(4),\ln(3),\ln(2))b$ and the threshold is $t=0.5$.
			Only product $1$ dominates product $4$.}
		\label{fig:Threshold2}
	\end{figure}
\end{example}

It is also necessary to update the definition of $\rho$  in Definition \ref{def:GLM},
since it now depends on the price of all products in the assortment.
The definition of $\rho: X\cup \left\{0\right\} \times 2^X\times \left(\mathbb{R}_+\cup \infty\right)^n \to [0,1]$ becomes:
\begin{equation}
\label{T_pbb}
\rho(i,S,p)=\begin{cases}
\frac{a_i(p_i)}{\sum_{j\in c(S,p)} a_j(p_j) + a_0}, & \text{if } i\in c(S,p), \\
0 & \text{if } i\notin c(S,p).
\end{cases}
\end{equation}
where $a_0$ is the attractiveness of the outside option.

\noindent
The expected revenue (ER) of an assortment $S\subseteq X$ and a
price vector $p\in \mathbb{R}_+^n$ is given by

\begin{equation}\label{ER_price}
\tag{ER}
R(S,p)=\sum_{i\in c(S,p)}\rho(i,S,p)p_i.
\end{equation}

A pair $(S,p)$ with $S\subseteq X$ and $p\in \left(\mathbb{R}_+\cup
{\infty}\right)^n$ is \textit{valid} if $S=\{i:p_i<\infty\}$ and
$c(S,p)=S$. Let $\mathcal{V}$ be the set of all valid pairs $(S,p)$.
Observe that one can always restrict the search for optimal solutions
to $\mathcal{V}$. Indeed, all dominated products can be given an   
infinite price and removing them from the original
assortment yields the exact same revenue.


%

The Joint Assortment and Pricing problem aims at finding a set $S^*$ and a price vector $p^*$ satisfying
\[
(S^*,p^*)\in\argmax_{(S,p)\in \mathcal{V}} R(S,p)
\]
and yielding an optimal revenue of
\[
R^*=R(S^*,p^*).
\]

First observe that the strategy used to solve this problem under the
multinomial logit does not carry over to the Threshold Luce Model.
Under the multinomial logit, the optimal solution for the joint
assortment and pricing problem is a fixed adjusted margin policy
\citep{wang2012} which, for equal price sensitivities and normalised
costs, translates to a fixed price policy.  As shown in
\cite{li2011pricing}, the optimal solution for the pricing problem
under the multinomial logit can be expressed in closed form using the
Lambert function $W(x):[0,\infty)\to [0,\infty)$ which is defined
as the unique function satisfying:
\begin{equation}\label{ec:Lambert}
x=W(x)e^{W(x)} \quad \forall x\in [0,\infty).
\end{equation}

\noindent
Using this function, the optimal revenue can be expressed as:

\begin{equation}\label{eq:R-interior}
R^*=W\left(\frac{\sum_{i\in X}\exp(u_i-1)}{a_0}\right)
\end{equation}

\noindent
The prices are all equal and satisfy: $p_i=1+R^*\quad \forall i\in
X$. 
The following example shows that fixed-price policy is not optimal
under the Threshold Luce Model.

\begin{example}[Fixed-Price policy is not optimal]\label{exmp:fixed_price_num}
	
	Consider $11$ products with product 1 having utility $u=2$ and all
	remaining $10$ products having utility $u'=1$.  Consider $a_0=1$ and
	$t=1$.  Observe that, for any fixed price, product $1$ always
	dominates the other $10$ products having lower utility, as $\exp(u-u')=
	\exp(1)=e > (1+t)=2$.  Therefore, the optimal revenue for a
	a fixed price strategy is:
	\[
	R_{fixed}=W\left(\frac{\exp(u-1)}{a_0}\right) = W(e) =1.
	\]
	As a result, the $10$ lower utility products are completely ignored
	and only product 1 contributes to the revenue.
	
	Consider the following price scheme now: let the price for product
	$1$ be $p=1.8$ and let the price be $p'=1.4$ for the remaining
	products. Product $1$ does not dominate any other product now. Indeed,
	for any $1<k\leq 11$,
	\[
	\frac{a_1}{a_k}=\exp((u-p)-(u'-p')) = \exp((2-1.8)-(1-1.4)) \approx 1.822 <1+t =2,
	\]
	which yields a revenue of:
	\begin{equation*}
	R'=\frac{p\cdot\exp(u-p) +10\cdot p'\exp(u'-p')}{\exp(u-p) +10\cdot\exp(u'-p')+a_0}= \frac{1.8\cdot\exp(2-1.8) +10\cdot 1.4\exp(1-1.4)}{\exp(2-1.8) +10\cdot\exp(1-1.4)+1}\approx 1.298,
	\end{equation*}
	This pricing scheme improves upon the fixed-price policy, yielding a revenue almost $\%30$ higher.
\end{example}

The intuition behind this example is as follows: For a fixed price
strategy, the only factor affecting dominance is the intrinsic
utilities because the prices vanish when calculating the ratio between
two attractiveness. This means that the solution can potentially miss the
benefits of low attractiveness products which are dominated by
the most attractive product.

It is thus important to understand the structure of an optimal
solution for the Joint Assortment and Pricing problem under the
Threshold Luce model. The first result states that, for any optimal
solution $(S^*,p^*)$, all product prices are greater or equal than
$R^*$, where $R^*$ denotes the revenue achieved at optimality.

\begin{proposition}\label{prop:price_bound}
	In any optimal solution $(S^*,p^*)$, for all $i\in S^*$, $p^*_i\geq R^*$.
\end{proposition}

The proof is by contradiction: Removing products with a price lower
than $R^*$ yields a greater revenue. The next proposition
characterises the optimal assortment of products of any optimal solution to
the Joint Assortment and Pricing problem. Recall that the products are
indexed by decreasing utility $u_i$. Thus, the set of products $\left[k\right] :=\left\{1,\ldots,k\right\}$, (with $0<k\leq n$) is said to be an \textit{intrinsic utility ordered set}.  The following proposition holds:

\begin{proposition}\label{prop:utility_order}
	Let $(S^*,p^*)$ denote an optimal solution. Then $S^*=\left[k\right]$ for some $k\leq n$.
\end{proposition}

The following Lemma due to \cite{wang2017impact} is useful to prove
some of the upcoming propositions. For completeness, its proof is also
in Appendix~\ref{App:proofs}.

\begin{lemma}[Lemma 1, \cite{wang2017impact}]\label{flat_p}
	Let $H(p_i,p_j):=p_i\cdot \exp(u_i-p_i) +p_j\cdot \exp(u_j-p_j)$, where $ \exp(u_i-p_i) + \exp(u_j-p_j)=T$.
	Then, $H(p_i,p_j)$ is strictly unimodal with respect to $p_i$ or $p_j$,
	and it achieves the maximum at the following point:
	\begin{equation}\label{opt_2}
	p_i^*=p_j^*=\ln\left((\exp(u_i)+\exp(u_j))/T\right)
	\end{equation}
\end{lemma}

Observe that setting the price of a product to $\infty$ is equivalent
to not showing it to consumers.  By Proposition
\ref{prop:utility_order}, one can always find an optimal solution that
is \textit{intrinsic utility ordered}.  Given a price vector $p\in
\mathbb{R}^n$, let $\gamma(p):\mathbb{R}^n\to \left[n\right]$ be
defined as $\gamma(p)\doteq \left\{\max_{i\in \left[n\right]} i \text{
	s.t } p_i<\infty\right\}$.  Intuitively, this is the last
non-infinite price. Proposition \ref{prop:decreasing_prices} shows
that, at optimality, the finite prices are non-increasing in $i$,
meaning that lower prices are assigned to lower utility products.

\begin{proposition}\label{prop:decreasing_prices}
	The prices at an optimal solution $(S^*,p^*)$ satisfy $p_{i}^*\geq p_{i+1}^* \quad \forall i\in \left[\gamma(p)-1\right]$.
	Moreover, if $i,j\in S^*$ satisfy $u_i=u_j$, then $p_i^*=p^*_j$.
\end{proposition}

Recall that the \textit{net utility} of product $i$ was defined as: $v_i=u_i-p_i$.
The following proposition shows that at optimality,
net utility follows the same order as intrinsic utility.

\begin{proposition}\label{prop:decreasing_net}
	Let $p^*$ be the price of an optimal solution of the Joint Assortment and Pricing Problem. The following condition holds: $u_i-p_i^*\geq u_{i+1}-p_{i+1}^* \quad \forall i\in \left[\gamma(p)-1\right]$.
\end{proposition}

The above propositions make it possible to filter out non-efficient
assortments and prices by restricting the search space to intrinsic
utility ordered assortments and providing insights on how the optimal
solution behaves regarding prices and their relation with utilities.
Based on these propositions, the joint assortment and pricing
optimisation problem for the TLM can be written in a more succinct
way. From Proposition \ref{prop:utility_order}, the solution is an
intrinsic utility ordered set $S_k=\left[k\right]$ for some $k\leq
n$. Suppose there exists an optimal solution in the form $(S_k,p)$ for
a fixed value $k$. In that case, recall that it is sufficient to restrict to
valid pairs $(S_k,p)$, meaning that $c(S_k,p)=S_k$. Consider a fixed
$k\leq n$. By Proposition \ref{prop:decreasing_net}, at optimality,
$u_i-p_i \geq u_j-p_j\quad \forall 1\leq i<j\leq k$. Therefore,
the condition that $c(S_k,p)=S_k$ can be written as

\begin{equation}\label{condition_valid_pairs_s_k}
g_{ij}(p) := \exp(u_i-p_i)-(1+t)\cdot \exp(u_j-p_j) \leq 0, \quad \forall 1\leq i<j\leq k
\end{equation}

As a result, the joint $k$-assortment and pricing optimisation problem for
the TLM (\texttt{JAPTLM-k}), which aims at finding an optimal assortment $S_k$
of size $k$ with $k\leq n$, can be written as:

\begin{equation}
\label{eq:max_R}
\tag{\texttt{JAPTLM-k}}
\begin{aligned}
& \underset{p}{\text{maximize}}
& & R^{(k)}(p) :=\frac{\sum_{i\in S_k}p_i\cdot\exp(u_i-p_i)}{\sum_{i\in S_k}\exp(u_i-p_i)+a_0}\\
& \text{subject to}
& & g_{ij}(p) \leq 0, \quad \forall 1\leq i<j\leq k   
\end{aligned}
\end{equation}

\noindent

Note that, if $\exp(u_1-u_k)\leq (1+t)$, then the solution is the same
as the unconstrained case, because any fixed price can be assigned
without creating dominances.  Hence, the optimal revenue
$\bm{R^{(k)}}$ can be calculated using equation~\eqref{eq:R-interior},
and all prices are equal to $1+\bm{R^{(k)}}$. On the other hand, if
$\exp(u_1-u_k)>1+t$, as in Example \ref{exmp:fixed_price_num}, the
prices need to be adjusted in order to avoid dominances.

The next theorem is the main result of this section.

\begin{theorem}\label{thm:ThrehsoldOpt}
	Problem \ref{eq:max_R}  can be solved in polynomial time.
\end{theorem}

The intuition behind the proof is based on Proposition
\ref{prop:decreasing_net} and the study of the Lagrangean relaxation
of problem~\eqref{eq:max_R}.  Observe that, since $u_i-p_i\geq
u_j-p_j\quad (i\leq j)$ at optimality, then the largest ratio
between attractiveness is obtained for products $1$ and $k$. This
ratio can also occur for more products but
only if they have the same net utility as products $1$ or
$k$.  Thus, it must be the case that there are non-negative integers
$k_1$ and $k_2$ with $k_1+k_2\leq k$, such that letting
$I_1=\left[k_1\right]$ and $I_2=\{k-k_2+1,k-k_2+2,\ldots,k\}$, the set
of constraints $C(k_1,k_2)=\{g_{ij}(p) \mid i\in I_1, j\in I_2 \}$ are
satisfied at equality for the optimal solution (see the proof in
Appendix~\ref{App:proofs} for details). Since it is only necessary to study a
polynomial number of combinations of constraints satisfied at
equality and, for each one of those combinations a closed form
solution is provided, the result follows.

For the non-trivial case with $\exp(u_1-u_k)>1+t$, where a fixed price
fails to be optimal, the prices need to be adjusted in order to avoid
the dominances. Let $\bm{R^{(k)}}$ and $\bm{p^{(k)}}$ be the optimal
revenue and price vector.  The following Lemma characterizes the
structure of the optimal solution for problem \ref{eq:max_R}.

\begin{lemma}\label{lem:sol_structure}
	The optimal solution to problem \eqref{eq:max_R} is either the same
	as the unconstrained case (i.e. fixed price, in the case that $\exp(u_1-u_k)\leq (1+t)$) or the following
	holds at optimality:
	
	\begin{equation}\label{tight_sol}
	\frac{a_1(p_1)}{a_k(p_k)}=1+t.
	\end{equation}
	
	Moreover, there are non-negative integers $k_1^*,k_2^*$, with $k_1^*+k_2^*\leq k$ such that:
	\begin{equation*}\label{eq:R^{(k)}_opt_ext}
	\bm{R^{(k)}}=W\left(\frac{\left(k_1^*+\frac{k_2^*}{1+t}\right)\cdot \exp\left(\frac{(1+t)\sum_{i\in I_1}u_i +\sum_{i\in I_2}u_i+k_2^*\ln(1+t)}{k_1^*(1+t)+k_2^*}-1\right) +\sum_{i\in \bar{I}_k}\exp(u_i-1)  }{a_0}\right),
	\end{equation*}
	where $I_1=\left[k_1^*\right]$, $I_2=\{k-k_2^*+1,k-k_2^*+2,\ldots,k\}$ and $\bar{I}_k=\left[k\right]\setminus\left(I_1\cup I_2\right)$.
	The optimal prices can be obtained as follows:
	
	%
	\begin{equation}
	\label{eq:opt_prices_k_k1_k2}
	\bm{p^{(k)}_i}=\begin{cases}
	1+\bm{R^{(k)}}+u_i -\frac{(1+t)\sum_{i\in I_1}u_i +\sum_{i\in I_2}u_i+k_2^*\ln(1+t)}{k_1^*(1+t)+k_2^*} & \text{if } i\in I_1, \\
	1+\bm{R^{(k)}}+u_i -\frac{(1+t)\sum_{i\in I_1}u_i +\sum_{i\in I_2}u_i+k_2^*\ln(1+t)}{k_1^*(1+t)+k_2^*}+\ln(1+t) & \text{if } i\in I_2,\\
	1+\bm{R^{(k)}} & \text{if } i\in \bar{I}_k.	
	\end{cases}
	\end{equation}
\end{lemma}

%
%

Let \texttt{TLM-Opt($X,u,a_0,k$)} be the procedure to obtain the
optimal solution for problem \eqref{eq:max_R}. Using
\texttt{TLM-Opt($X,u,a_0,k$)} at most $n$ times (once for each $k\leq
n$) to obtain the assortment and prices yielding the highest
$\bm{R^{(k)}}$, one can find the optimal assortment and price vector
for any given instance.  Its intuition is to mimic the optimal
strategy for the regular MNL (Fixed-Price Policy) as much as possible.
However, given that it needs to accommodate prices in order to avoid
dominances, the algorithm adjusts prices for the higher intrinsic
utility products (making prices larger, hence less attractive) and
reduces the price of lower intrinsic utility ones, making them more
attractive for customers and preventing them from being dominated.
This allows the optimal strategy to have an edge over strategies
ignoring the Threshold induced dominances, such as Fixed-Price Policy
and, to a lesser extent, the Quasi-Same Price
\citep{wang2017impact}. The Quasi-Same Price policy policy only
adjusts the price of the lowest attractiveness product, instead of
adjusting both extremes of the attractiveness spectrum and potentially
multiple products.

\section{Conclusion and Future Work}
\label{sec:conclusion}

This paper studies the assortment optimization problem under the
Two-Stage Luce model (2SLM), a discrete choice model introduced by
\cite{echenique2018} that generalizes the standard multinomial logit
model (MNL) with a dominance relation and may violate regularity. The
paper proved that the assortment problem under the 2SLM can be solved
in polynomial time. The paper also considered the capacitated
assortment problem under the 2SLM and proved that the problem becomes
NP-hard in this setting.  We also provide polynomial-time algorithms
for special cases of the capacitated problem when (1) the dominance
relation is utility-correlated and when (2) its transitive reduction
is a forest. We also provide an Appendix showing
numerical experiments to highlight the performance of the proposed
algorithms against classical strategies used in the literature.

There are at least five interesting avenues for future
research. First, one may wish to study how to generalize the 2SLM
further while still keeping the assortment problem solvable in
polynomial time. For example, one can try to check whether there
exists a model that unifies the 2SLM and the elegant work in
\cite{davis2013assortment} where the assortment problem is still
solvable in polynomial time. Second, given that the capacitated
version of the 2SLM is NP-hard under Turing reductions (Theorem
\ref{np-hard-proof}), it is interesting to see whether there exist
good approximation algorithms for this problem. Third, one can explore
different forms of dominance. For example, one may consider dominances
specified by a discrete relation or a continuous functional form
between products.  Fourth, one can try to generalise our results for
the Joint Assortment Pricing Problem under the Threshold Luce model to
a more general setting, where price sensitivities depend on each
product. Finally, one can try to mix attention models with dominance
relations, meaning that a customer first perceives a subset of the
products, dictated by an attention filter, and then filter the
products even more using dominance relations.


	
	\section{Acknowledgements}
	
We thanks Yuval Filmus for his helpful insights leading us to find useful literature on this topic. Thanks are also due to Guillermo Gallego for suggesting extending our assortment results to the GAM model, and to Flavia Bonomo for relevant discussions..
	
	\bibliographystyle{../../../Bib/aaai}
	\bibliography{../../../Bib/biblio}
	
	\newpage
	\appendix

	\section{Proofs}\label{App:proofs}
	In this section we provide the proofs missing from the main text.
	
	\emph{Proof of Theorem~\ref{np-hard-proof}.}
	The proof considers four problems:
	\begin{enumerate}
		\item Problem (\texttt{MWBISBP}): Maximum weighted independent set of
		size at most $C$ for bipartite graphs.
		
		\item Problem (\texttt{MWEBISBP}): Maximum weighted independent set of
		size equal to $C$ for bipartite graphs.
		
		\item Problem (\texttt{EC2SLMAP}): Optimal assortment under the General
		Luce model of size $C$.
		
		\item Problem (\texttt{C2SLMAP}): Optimal capacitated assortment under
		the Two-Stage Luce model of size at most $C$.
	\end{enumerate}
	
	\noindent
	The proof shows that Problems (\texttt{MWEBISBP}), (\texttt{EC2SLMAP}), and
	(\texttt{C2SLMAP}) are NP-hard, using the NP-hardness of Problem
	(\texttt{MWBISBP}) \citep{Kalra2017} as a starting point.
	
	First observe that Problem (\texttt{MWEBISBP}) is NP-hard under Turing
	reductions. Indeed, Problem (\texttt{MWBISBP}) can be reduced to solving
	$C$ instances of Problem (\texttt{MWEBISBP}) with budget $c$ ($1 \leq c
	\leq C$).
	
	We now show that Problem (\texttt{EC2SLMAP}) is NP-hard. Consider Problem
	(\texttt{MWEBISBP}) over a bipartite graph $G=(V=V_1 \cup V_2,E)$, where
	$V_1 \cap V_2 = \emptyset$, every edge $(v_1,v_2) \in E$ satisfies
	$v_1 \in V_1$ and $v_2 \in V_2$, $w_v$ is the weight of vertex $v$,
	and $C$ is the budget. We show that Problem (\texttt{MWEBISBP}) over this
	bipartite graph can be polynomially reduced to Problem
	(\texttt{EC2SLMAP}). The reduction assigns each vertex $v$ to a product
	with $a(v) = 1$ and $r_v = w_v$, sets $a_0 = 0$, and has a capacity
	$C$. Moreover, the reduction uses the following dominance relation:
	$v_1 \succ v_2$ iff $(v_1,v_2) \in E$. This dominance relation is
	irreflexive, anti-symmetric, and transitive, since the graph is
	bipartite. A solution to Problem (\texttt{MWEBISBP}) is a feasible
	solution to Problem (\texttt{EC2SLMAP}), since the independent set cannot
	contain two vertices $v_1, v_2$ with $v_1 \succ v_2$ by
	construction. Similarly, a feasible assortment is an independent set,
	since the assortment cannot select two vertices $v_1 \in V_1$ and $v_2
	\in V_2$ with $(v_1,v_2) \in E$, since $v_1 \succ v_2$. The objective
	function of Problem (\texttt{EC2SLMAP}) reduces to maximizing
	\[
	\frac{1}{C} \sum_{v \in V} r_v x_v
	\]
	which is equivalent to maximizing $\sum_{v \in V} r_v x_v$ since
	exactly $C$ products will be selected by every feasible
	assortment. The result follows by the NP-hardness of Problem
	(\texttt{MWEBISBP}).
	
	Finally, Problem (\texttt{C2SLMAP}) is NP-hard under Turing
	reductions. Indeed, Problem (\texttt{C2SLMAP}) can be reduced to solving
	$C$ instances of Problem (\texttt{EC2SLMAP}) with capacity $c$ ($1 \leq c
	\leq C$).\qed
	
	\emph{Proof of Theorem~\ref{thm:forest}}
By Theorem \ref{rational}, it suffices to show that Problem
\eqref{MACC} is solved by the recurrences in polynomial time. The
correctness of recurrence $\mathcal{A}(v,c)$ comes from the fact
that vertex $v$ dominates all its descendants and cannot be present
in any assortment featuring any of them. The correctness of
recurrence $\mathcal{A}^+(S,c)$ follows from the fact that $e$ is
not dominated by, and does not dominate, any element in $S$, since
they are all children of the same node. This also holds for the
descendants of $e$ and the descendants of the elements in
$S$. Hence, the optimal assortment is obtained by splitting the
capacity $c$ into $n_1$ and $n_2$ and merging the best assortment
for $\mathcal{A}^+(S,n_1)$ and $\mathcal{A}(e,n_2)$ for some $n_1,
n_2 \geq 0$ summing to $c$.  The recurrences can be solved in
polynomial time since the computation for each vertex $v$ and
capacity $c$ takes $O(n \ C)$ time, giving an overall time
complexity of $O(n^2 \ C^2)$.\qed

%

\emph{Proof of Proposition \ref{prop:price_bound}.}
We prove this by contradiction. Suppose $p_i^*<R^*$ for some $i\in S$, then $\hat{S}=S^*\setminus\{i\}$
has better revenue than the optimal solution if we keep the same prices and $p_i^*<R^*$. Indeed, let us calculate $R(\hat{S})$:
\begin{align*}
R(\hat{S})=& \frac{\sum_{j\in \hat{S}} e^{u_j-p_j^*}\cdot p_j^*}{\sum_{j\in \hat{S}} e^{u_j-p_j^*}+a_0}\\
R(\hat{S})=& \frac{\sum_{j\in S^*} e^{u_j-p_j^*}\cdot p_j^*-e^{u_i-p_i^*}\cdot p_i^*}{\sum_{j\in S^*} e^{u_j-p_j^*}-e^{u_i-p_i^*}+a_0}\\
R(\hat{S})=& \frac{\sum_{j\in S^*} e^{u_j-p_j^*}\cdot p_j^*}{\sum_{j\in S^*} e^{u_j-p_j^*}+a_0}\cdot\frac{\sum_{j\in S^*} e^{u_j-p_j^*}+a_0}{\sum_{j\in S^*} e^{u_j-p_j^*}-e^{u_i-p_i^*}+a_0} -\frac{ e^{u_i-p_i^*}\cdot p_i^*}{\sum_{j\in S^*} e^{u_j-p_j^*}-e^{u_i-p_i^*}+a_0}\\
R(\hat{S})=& \frac{\sum_{j\in S^*} e^{u_j-p_j^*}\cdot p_j^*}{\sum_{j\in S^*} e^{u_j-p_j^*}+a_0}\cdot\left[1+\frac{e^{u_i-p_i^*}}{\sum_{j\in S^*} e^{u_j-p_j^*}-e^{u_i-p_i^*}+a_0}\right] -\frac{ e^{u_i-p_i^*}\cdot p_i^*}{\sum_{j\in S^*} e^{u_j-p_j^*}-e^{u_i-p_i^*}+a_0}\\
R(\hat{S})=& R^*\cdot\left[1+\frac{e^{u_i-p_i^*}}{\sum_{j\in S^*} e^{u_j-p_j^*}-e^{u_i-p_i^*}+a_0}\right] -\frac{ e^{u_i-p_i^*}\cdot p_i^*}{\sum_{j\in S^*} e^{u_j-p_j^*}-e^{u_i-p_i^*}+a_0}\\
R(\hat{S})=& R^* +\underbrace{\frac{ e^{u_i-p_i^*}}{\sum_{j\in S^*} e^{u_j-p_j^*}-e^{u_i-p_i^*}+a_0}\cdot \left[R^*-p_i^*\right]}_{\Gamma}
\end{align*}

Now $\Gamma$ is positive because $p_i^*<R^*$, but this implies $R(\hat{S})>R^*$, contradicting the optimality of $R^*$.\qed

\emph{Proof of Proposition~\ref{prop:utility_order}.}
Let $(S^*,p^*)$ be an optimal solution. We can assume that $(S^*,p^*)\in\mathcal{V}$.
We proceed by contradiction. Suppose that there is a product $i$ not included in
the optimal solution and another product $j$ with \textbf{smaller intrinsic utility}
included in $S^*$. We show that we can include product $i$, and remove $j$ and get a greater revenue.
Let $\hat{S}=(S^*\setminus\{j\})\cup \{i\}$, be the set where we removed product $j$, and included product $i$.
Let $\hat{p}_i=u_i-u_j+p_j^*$, this means that the total attractiveness remains unchanged,
and no new domination relations appear, given that product $j$ already had
the same level attractiveness that product $i$ now has.
Observe that given that $u_i\geq u_j$, we have that $\hat{p}_i\geq p_j^*$.
Let us calculate $R(\hat{S},\hat{p})$, where $\hat{p}$ is the same as $p^*$,
but with the proposed changes in price:

\begin{align*}
R(\hat{S},\hat{p})&= \frac{\sum_{k\in \hat{S}} e^{u_k-\hat{p}_k}\cdot \hat{p}_k}{\sum_{k\in \hat{S}} e^{u_k-\hat{p}_k}+a_0}\\
R(\hat{S},\hat{p})&=\frac{\sum_{k\in S^*} e^{u_k-p_k^*}\cdot p_k^* - e^{u_j-p_j^*}\cdot p_j^* +e^{u_i-\hat{p}_i}\cdot \hat{p}_i}{\sum_{k\in \hat{S}} e^{u_k-\hat{p}_k}+a_0}\\
R(\hat{S},\hat{p})&=\underbrace{\frac{\sum_{k\in S^*} e^{u_k-p_k^*}\cdot p_k^*}{\sum_{k\in S^*} e^{u_k-\hat{p}_k}+a_0}}_{R^*} + \frac{e^{u_i-\hat{p}_i}\cdot \hat{p}_i -e^{u_j-p_j^*}\cdot p_j^*}{\sum_{k\in \hat{S}} e^{u_k-\hat{p}_k}+a_0}\\
R(\hat{S},\hat{p})&=R^* + \underbrace{\frac{e^{u_j-p_j^*}}{\sum_{k\in \hat{S}} e^{u_k-\hat{p}_k}+a_0}}_{\geq 0}\cdot \underbrace{\left[\hat{p}_i-p_j^*\right]}_{> 0}\\
R(\hat{S},\hat{p})&> R^*
\end{align*}

Where we first rewrite $R(\hat{S},\hat{p})$ using $(S^*,p^*)$ because we just swapped product $i$ for product $j$,
and the total attractiveness remain the same, so the denominator does not change.
Then we identify $R(S,p)$, and we use $u_i-\hat{p}_i=u_j-p_j$ to being able to factorize the remaining terms.
So we found a pair $(\hat{S},\hat{p})$, yielding strictly more revenue than $(S,p)$,
but adding product $i$, which contradicts the optimality of $(S^*,p^*)$.\qed

\emph{Proof of Lemma~\ref{flat_p}.}
The proof (due to \cite{wang2017impact}) is useful because it provides intuition on how the optimal price variates
when constrained to a fixed additive market share among any two products. By the equality constraint,
we have $p_j=u_j-\ln(T-\exp(u_i-p_i))$, so $H(p_i,p_j)$ can be rewritten purely as a function of $p_i$ as:
\begin{equation}\label{H_pi}
H(p_i)=p_i\cdot \exp(u_i-p_i) +(u_j-\ln(T-\exp(u_i-p_i)))\cdot (T-\exp(u_i-p_i)).
\end{equation}

Now, let us calculate the first derivative of $H(p_i)$ w.r.t. $p_i$:

\begin{equation}\label{dH}
\frac{\partial H(p_i)}{\partial p_i}=\left(-p_i + (u_j-\ln(T-\exp(u_i-p_i)))\right)\cdot\exp(u_i-p_i)
\end{equation}
Clearly the left-hand side term on the multiplication is monotonically decreasing from
positive to negative values as $p_i$ increases from $0$ to $\infty$.
Therefore $H(p_i)$ is strictly unimodal and reaches its maximum value at:

\[	p_i^*=p_j^*=\ln\left((\exp(u_i)+\exp(u_j))/T\right).\] \qed

\emph{Proof of Proposition~\ref{prop:decreasing_prices}.}
We prove this result by contradiction. Let $i$ be the first index where this condition does not hold,
this means that $p_i^*<p_{i+1}^*$. Using Lemma \ref{flat_p}, we found $\hat{p}$ satisfying $p_i^*<\hat{p}<p_{i+1}^*$.
Does this new price alter the consideration set? We show that this is not the case.
Indeed, the effect is two-fold: the price for product $i$ increases, and the price for product $i+1$ decreases.
We analyse the effect of these two consequences:

\begin{itemize}
	\item Increase on price for product $i$: This means $a(i,p)$ decreases. Note that $u_i-\hat{p}\geq u_{i+1}-p_{i+1}^*$,
	so neither $i\succ {i+1}$ or ${i+1}\succ i$, because their attractiveness are now even closer than before.
	Can $i$ be dominated now by another product? No, because given that $u_i\geq u_{i+1}$ we have $u_i-\hat{p}\geq u_{i+1}-\hat{p}\geq u_{i+1}-p_{i+1}^*$.
	Therefore the new attractiveness of $i$ is still larger than the new attractiveness of ${i+1}$,
	and the last inequality implies that the new attractiveness of $i$ is larger than the old attractiveness of ${i+1}$,
	and ${i+1}$ was not previously dominated either by any other product.
	
	\item Decrease on price for product $i+1$: Previously ${i+1}$ was not dominated by any product.
	Can ${i+1}$ be dominated now? No, because if ${i+1}$ was not dominated before,
	now with a smaller price $\hat{p}$ its attractiveness is larger and therefore can't be
	dominated now either (the only other product that changed attractiveness was $i$,
	and it now has smaller attractiveness). Can ${i+1}$ dominate another product now with its new higher attractiveness?
	No, because given that $u_i\geq u_{i+1}$ we have $u_i-p_i^*\geq u_{i+1}-p_i^*\geq u_{i+1}-\hat{p}$,
	so the old attractiveness of product $i$ is larger than the new attractiveness of product ${i+1}$,
	and given that $i$ did not dominate another product before, the new price does not make ${i+1}$ dominate another product either.
\end{itemize}

So, letting $p^{fix}$ exactly the same as $p^*$, but replacing both $p_i^*$ and $p_{i+1}^*$ with $\hat{p}$,
means that the pair $(S^*,p^{fix})$ yields strictly more revenue than $(S^*,p^*)$ (by Lemma \ref{flat_p}),
contradicting the optimality assumption.
The fact that equal intrinsic utility implies equal price at optimality, can be easily demonstrated by the following:
if two equal intrinsic utility products have different prices, then using Lemma \ref{flat_p}
we obtain strictly better revenue by assigning them the same price, and no new domination occurs,
because the new price is confined between the previous prices.\qed

\emph{Proof of Proposition~\ref{prop:decreasing_net}.}
We prove this by contradiction. Let $p^*$ be the optimal solution and $i$ be the first index where this condition does not hold.
This means that $u_i-p_i^*< u_{i+1}-p_{i+1}^*$. We can extrapolate this inequality further and say:
\begin{equation}\label{ineq_chain}
u_{i+1}-p_i^* <u_i-p_i^*< u_{i+1}-p_{i+1}^*<u_i-p_{i+1}^*,
\end{equation}
because $u_i\geq u_{i+1}$  and $p_i\geq p_{i+1}$  by Propositions  \ref{prop:utility_order} and \ref{prop:decreasing_prices} respectively.
We now do the following: Define $p_{i}'$ and $p_{i+1}'$ such as $\exp(u_i-p_i')+\exp(u_{i+1}-p_{i+1}')=\exp(u_i-p_i^*)+\exp(u_{i+1}-p_{i+1}^*)$
and $\exp(u_i-p_i')=\exp(u_{i+1}-p_{i+1}')$. This means that:

\begin{align*}\label{new_price}
p_i'&=u_i-\ln\left(\frac{\exp(u_i-p_i^*)+\exp(u_{i+1}-p_{i+1}^*)}{2}\right) \\
p_{i+1}'&=u_{i+1}-\ln\left(\frac{\exp(u_i-p_i^*)+\exp(u_{i+1}-p_{i+1}^*)}{2}\right)
\end{align*}

Consider  $H(p_i,p_{i+1})=p_i\cdot \exp(u_i-p_i) +p_j\cdot \exp(u_{i+1}-p_{i+1})$,
where $ \exp(u_i-p_i) + \exp(u_i-p_{i})=\exp(u_i-p_i^*)+\exp(u_{i+1}-p_{i+1}^*)$.
By Lemma \ref{flat_p}, $H(p_i,p_{i+1})$ is strictly increasing in $p_i$ for $p_i\leq\hat{p}$
and strictly decreasing for $p_i\geq \hat{p}$, with $\hat{p}=\ln\left(\frac{\exp(u_i)+\exp(u_{i+1})}{\exp(u_i-p_i^*)+\exp(u_{i+1}-p_{i+1}^*)}\right)$
the solution of the corresponding maximization problem of Lemma \ref{flat_p}.
We can verify that $\hat{p}<p_i'<p_i^*$. The first inequality is straightforward. Indeed:

\begin{align*}
p_i'=&u_i-\ln\left(\frac{\exp(u_i-p_i)+\exp(u_{i+1}-p_{i+1})}{2}\right)\\
p_i'&=\ln\left[\frac{2\exp(u_i)}{\exp(u_i-p_i^*)+\exp(u_{i+1}-p_{i+1}^*)}\right]\\
p_i'&>\underbrace{\ln\left[\frac{\exp(u_i)+\exp(u_{i+1})}{\exp(u_i-p_i^*)+\exp(u_{i+1}-p_{i+1}^*)}\right]}_{\hat{p}}\\
p_i'>&\hat{p}
\end{align*}
proving the desired inequality. Now, for the second one:

\begin{align*}
p_i'&=u_i-\ln\left(\frac{\exp(u_i-p_i)+\exp(u_{i+1}-p_{i+1})}{2}\right)\\
p_i'&=\ln\left[\frac{2\exp(u_i)}{\exp(u_i-p_i^*)+\exp(u_{i+1}-p_{i+1}^*)}\right]\\
p_i'&\leq \ln\left[\frac{2\exp(u_i)}{\exp(u_i-p_i^*)+\exp(u_{i}-p_{i+1}^*)}\right]\\
p_i'&= \ln\left[\frac{2\exp(u_i)}{\exp(u_i)(\exp(-p_i^*)+\exp(-p_{i+1}^*))}\right]\\
p_i'&< \ln\left[\frac{2}{2\exp(-p_i^*)}\right]\\
p_i'&< p_i^*,
\end{align*}

thus we have:

\begin{equation*}\label{conclude_rev}
p_i'\cdot\exp(u_i-p_i')+p_{i+1}'\cdot\exp(u_{i+1}-p_{i+1}')>p_i^*\cdot\exp(u_i-p_i^*)+p_{i+1}^*\cdot\exp(u_{i+1}-p_{i+1}^*).
\end{equation*}

Meaning that we have the same assortment, but with prices $p_{i}'$ and $p_{i+1}'$ generating strictly more revenue than the optimal prices,
which is a contradiction. The only thing that we have left to show that with these new prices we are still on the same consideration set.
It would be enough to show that the new net utilities are bounded by previous values of net utilities.
Indeed, we can verify that $p_{i+1}^*\leq p_{i+1}'\leq p_i'\leq p_i^*$, by simply using the definitions.
We also know, by hypothesis that $u_i-p_i'=u_{i+1}-p_{i+1}'$, then $u_i-p_i'=u_{i+1}-p_{i+1}'\leq u_{i+1}-p_{i+1}^*$.
So even when the price of product $i$ decreased, the new attractiveness is bounded above by a previously existing attractiveness,
thus not changing the consideration set. By the same reasoning, $u_{i+1}-p_{i+1}'=u_i-p_i'\geq u_i-p_i^*$,
meaning that the new attractiveness is bounded below by a pre-existing one, so $i+1$ is not dominated with this new prices either.
So the consideration set stays the same, concluding the proof.\qed
\endproof

\emph{Proof of Theorem~\ref{thm:ThrehsoldOpt}.}
We first write problem \eqref{eq:max_R} in minimization form to directly apply the Karush-Khun-Tucker conditions (KKT)\citep{Karush39}.

\begin{equation}
\label{eq:min_R}
\begin{aligned}
& \underset{p}{\text{minimize}}
& & -R^{(k)}(p) \\
& \text{subject to}
& & g_{ij}(p) \leq 0, \quad \forall 1\leq i<j\leq k   
\end{aligned}
\end{equation}
The associated Lagrangean function is:

\begin{equation}\label{eq:L_k}
\mathcal{L}_k(p,\mu)=-R^{(k)}(p) +\sum_{1\leq i<j\leq k }\mu_{ij}\cdot g_{ij}(p),  
\end{equation}
\noindent
where $\mu_{ij}\geq 0$ are the associated Lagrange multipliers.
Recall that if $\exp(u_1-u_k)\leq(1+t)$,
the optimal revenue $\bm{R^{(k)}}$ can be calculated using
equation \eqref{eq:R-interior}, and the solution corresponds
to a fixed price policy as for the regular multinomial logit.

On the other hand, 	if $\exp(u_1-u_k)>(1+t)$, any fixed price causes product
$k$ to be dominated by product $1$. Thus, to include product $k$
in the assortment we need to adjust the prices.
Let $p=\left(p_1,\ldots,p_k\right)$ be the
optimal price vector for problem \eqref{eq:min_R}.
Observe that it can't happen that $\frac{a_1(p_1)}{a_k(p_k)}<1+t$,
since by Proposition \ref{prop:decreasing_net},
it will also means that $\frac{a_1(p_1)}{a_2(p_2)}<1+t$ and	
using Lemma \ref{flat_p} we can find $\hat{p}$ such that assigning
$\hat{p}$ to products $1$ and $2$ yields a larger revenue (and no dominance relation
appears, since the attractiveness of product $1$ was reduced, and the attractiveness of product $2$
increased, but is still less than the one of product $1$), which contradicts optimality.
Therefore, $g_{1k}$ must be satisfied with equality, meaning $\frac{a_1(p_1)}{a_k(p_k)}=1+t$.

Furthermore, at optimality it holds $u_i-p_i\geq u_j-p_j\quad\forall i\leq j$
(by Proposition \ref{prop:decreasing_net}), and thus
the biggest ratio between attractiveness is observed
for products $1$ and $k$, and is exactly equal to $1+t$.
This ratio can be replicated for other pairs of products,
but only if they share the same net utility (and thus attractiveness) to the one of products $1$ or $k$.
Therefore, it must be the case that there are integers $k_1$ and $k_2$ with $k_1+k_2\leq k$, such that
all products in $I_1=\left[k_1\right]$ share the same attractiveness ($a_1(p_1)$)
and all products in $I_2=\{k-k_2+1,k-k_2+2,\ldots,k\}$ share the same attractiveness as well ($a_k(p_k)$).
This means that the set of constraints $C(k_1,k_2)=\{g_{ij}(p) \mid i\in I_1, j\in I_2 \}$ are
all  satisfied	with equality at optimality.

We now study the derivative of equation \eqref{eq:L_k} with respect to each price $p_i$ to obtain the KKT conditions.
We here assume that the first $k_1$ values share the same net utility value,
meaning $u_s=u_1-p_1=u_i-p_i \quad \forall i \in I_1$, and for the last $k_2$ products,
we also have the same value of net utility, that we call $u_f$, this is:
$u_f=u_k-p_k=u_i-p_i\quad \forall i\in I_2$. Where these two quantities satisfy:

\[u_s-u_f=\ln(1+t),\]

Let us write the derivatives of the Lagrangean depending on where the index $i$ belongs. If $i\in I_1$, then:

\begin{equation}\label{KKT_p_k1}
\frac{dL_k}{dp_i}= \frac{\exp(u_i-p_i)}{\sum_{j\in S_k}\exp(u_j-p_j)+a_0}\cdot\left[p_i-1-R^{(k)}(p)\right] - \exp(u_i-p_i)\cdot\sum_{j\in I_2} \mu_{ij},
\end{equation}
if $i\in I_2$, we have:
\begin{equation}\label{KKT_p_k2}
\frac{dL_k}{dp_i}= \frac{\exp(u_i-p_i)}{\sum_{j\in S_k}\exp(u_j-p_j)+a_0}\cdot\left[p_i-1-R^{(k)}(p)\right] + (1+t)\exp(u_i-p_i)\cdot\sum_{j\in I_1} \mu_{ji},
\end{equation}
And finally, if $i\in \bar{I}_k=\left[k\right]\setminus\left(i_1\cup I_2\right)$, the derivative takes the following form:
\begin{equation}\label{KKT_p_o}
\frac{dL_k}{dp_i}= \frac{\exp(u_i-p_i)}{\sum_{j\in S_k}\exp(u_j-p_j)+a_0}\cdot\left[p_i-1-R^{(k)}(p)\right]
\end{equation}

Observe that $\forall i\in \bar{I}_k$, $\frac{dL_k}{dp_i}=0 \implies p_i=1+ R^{(k)}(p)$, and the right hand side is not dependent on $i$, so all products in $
\bar{I}_k$ share the same price, which we denote $\bar{p}$.  We can rewrite all prices and the revenue depending on $u_s$ and $\bar{p}$, using the following relations:

\begin{enumerate}
	\item $\forall i\in I_1 \quad u_1-p_1=u_i-p_i \implies p_i=u_i-u_s$
	\item $\forall i\in I_2 \quad u_1-p_1=u_i-p_i+\ln(1+t) \implies p_i=u_i-u_s+\ln(1+t)$
\end{enumerate}

Note now that at optimality, for a fixed $k$, prices are determined by $k_1$ and $k_2$. Thus, the optimal revenue can be written explicitly depending on $k$, $k_1$ and $k_2$, 
taking the following form:
\small
\begin{align}\label{eq:R1_prev}
&R^{(k)}(k_1,k_2)=\nonumber\\
&\frac{\sum_{i\in I_1}(u_i-u_s)\exp(u_s) + \bar{p}\exp(-\bar{p})\sum_{i\in\bar{I}_k}\exp(u_i) + \sum_{i\in I_2}(u_i-u_s+\ln(1+t))\exp(u_s-\ln(1+t)) }{\sum_{i\in I_1}\exp(u_s) + \exp(-\bar{p})\sum_{i\in\bar{I}_k}\exp(u_i) + \sum_{i\in I_2}\exp(u_s+\ln(1+t))+a_0}
\end{align}
\normalsize

Note that $\bar{p}=1+R^{(k)}(k_1,k_2)$ (Equation \eqref{KKT_p_o}) and let $E(k_1,k_2)=\sum_{i\in\bar{I}_k}\exp(u_i)$. Using these two relations, we can rewrite the optimal revenue as:
\small
\begin{equation}\label{eq:R1}
R^{(k)}(k_1,k_2)=\frac{e^{u_s}\sum\limits_{i\in I_1}(u_i-u_s) + \frac{e^{u_s}}{1+t}\cdot\sum\limits_{i\in I_2}(u_i-u_s+\ln(1+t))+E(k_1,k_2)(1+R^{(k)}(k_1,k_2))e^{-(1+R^{(k)}(k_1,k_2))}}{e^{u_s}\left[k_1+\frac{k_2}{1+t}\right]+E(k_1,k_2)e^{-(1+R^{(k)}(k_1,k_2))}+a_0}
\end{equation}
\normalsize

Up to this point, we have an equation relating the optimal revenue $R^{(k)}(k_1,k_2)$ and $u_s$.
From equation \eqref{KKT_p_k1}, after reordering terms we have:

\begin{align}\label{eq:kkt_1}
\frac{p_i-1-R^{(k)}(k_1,k_2)}{e^{u_s}\left(k_1+k_2(1+t)\right)+E(k_1,k_2)e^{-(1+R^{(k)}(k_1,k_2))}+a_0}	&=\sum_{j\in I_2} \mu_{ij},
\quad \quad \forall i\in I_1	\nonumber\\
\frac{u_i-u_s-1-R^{(k)}(k_1,k_2)}{e^{u_s}\left(k_1+k_2(1+t)\right)+E(k_1,k_2)e^{-(1+R^{(k)}(k_1,k_2))}+a_0}	&=\sum_{j\in I_2} \mu_{ij},
\quad \quad \forall i\in I_1
\end{align}

Analogously, from equation \eqref{KKT_p_k2}, after reordering terms we have $\forall i\in I_2$:

\begin{align}\label{eq:kkt_2}
\frac{p_i-1-R^{(k)}(k_1,k_2)}{e^{u_s}\left(k_1+k_2(1+t)\right)+E(k_1,k_2)e^{-(1+R^{(k)}(k_1,k_2))}+a_0}	&=-(1+t)\sum_{j\in I_1} \mu_{ji},
\quad \quad \forall i\in I_2	\nonumber\\
\frac{1}{1+t}\cdot\frac{u_i-u_s+\ln(1+t)-1-R^{(k)}(k_1,k_2)}{e^{u_s}\left(k_1+k_2(1+t)\right)+E(k_1,k_2)e^{-(1+R^{(k)}(k_1,k_2))}+a_0}	&=-\sum_{j\in I_1} \mu_{ji},
\quad \quad \forall i\in I_2
\end{align}

Now, if we add equations \eqref{eq:kkt_1} $\forall i\in I_1$ then take equations
\eqref{eq:kkt_2} and also add them $\forall i\in I_2$, and add those two results we can derive the value $R^{(k)}(k_1,k_2)$ as follows.

\begin{align}\label{eq:R2}
\underbrace{\sum_{\substack{i\in I_1\\j\in I_2}} \mu_{ij}-\sum_{\substack{i\in I_1\\j\in I_2}} \mu_{ij}}_{0}=&\frac{\sum_{i\in I_1}(u_i-u_s-1-R^{(k)}(k_1,k_2))+\frac{\sum_{i\in I_2}(u_i-u_s+\ln(1+t)-1-R^{(k)}(k_1,k_2))}{1+t}}{e^{u_s}\left(k_1+k_2(1+t)\right)+E(k_1,k_2)e^{-(1+R^{(k)}(k_1,k_2))}+a_0} \nonumber\\
R^{(k)}(k_1,k_2)\left(k_1+\frac{k_2}{1+t}\right)=&\sum_{i\in I_1}u_i + \frac{1}{1+t}\cdot\sum_{i\in I_2}u_i -(1+u_s)\cdot\left(k_1+\frac{k_2}{1+t}\right) +\frac{k_2\ln(1+t)}{1+t}\nonumber\\
R^{(k)}(k_1,k_2)=&\frac{(1+t)\sum_{i\in I_1}u_i +\sum_{i\in I_2}u_i+k_2\ln(1+t)}{k_1(1+t)+k_2}-1-u_s
\end{align}

We now have two equations relating $R^{(k)}(k_1,k_2)$ and $u_s$ in \eqref{eq:R1} and \eqref{eq:R2}.
Using these equations we can find the values of the optimal revenues and all the pricing structure
while varying $k_1$ and $k_2$. If we define the following constant:

\begin{equation}\label{eq:C1}
C_1(k_1,k_2)=\frac{(1+t)\sum_{i\in I_1}u_i +\sum_{i\in I_2}u_i+k_2\ln(1+t)}{k_1(1+t)+k_2}-1,
\end{equation}

Equation \eqref{eq:R2} becomes:

\begin{equation}\label{eq:R2_def}
R^{(k)}(k_1,k_2)=C_1(k_1,k_2)-u_s,
\end{equation}

and from Equation \eqref{eq:R2_def}, we can deduce the following relations:

\begin{equation}\label{eq:rels}
1+R^{(k)}(k_1,k_2)=C_1(k_1,k_2)-u_s+1 , \quad\quad \text{and} \quad\quad e^{-(1+R^{(k)}(k_1,k_2))}=e^{u_s-C_1(k_1,k_2)-1}.
\end{equation}

We will use these relations on Equation \eqref{eq:R1}. Let us first multiply both sides by the denominator on the right side:
\small
\begin{align}
&R^{(k)}(k_1,k_2) \cdot\left(e^{u_s}\left[k_1+\frac{k_2}{1+t}\right]+E(k_1,k_2)e^{-(1+R^{(k)}(k_1,k_2))}+a_0\right)\nonumber \\&=e^{u_s}\sum\limits_{i\in I_1}(u_i-u_s) + \frac{e^{u_s}}{1+t}\cdot\sum\limits_{i\in I_2}(u_i-u_s+\ln(1+t))+E(k_1,k_2)(1+R^{(k)}(k_1,k_2))e^{-(1+R^{(k)}(k_1,k_2))}\nonumber
\end{align}
\normalsize

using equations \eqref{eq:rels} to replace the value of $R^{(k)}(k_1,k_2)$ and write everything depending on $u_s$ we have:

\small
\begin{align}\label{eq:Start_R}
\left(C_1(k_1,k_2)-u_s\right)&\left(e^{u_s}\left[k_1+\frac{k_2}{1+t}\right]+E(k_1,k_2)e^{u_s-C_1(k_1,k_2)-1}+a_0\right)=e^{u_s}\sum\limits_{i\in I_1}(u_i-u_s) \nonumber\\ &+ \frac{e^{u_s}}{1+t}\cdot\sum\limits_{i\in I_2}(u_i-u_s+\ln(1+t))
+E(k_1,k_2)\cdot \left(C_1(k_1,k_2)-u_s-1\right)e^{u_s-C_1(k_1,k_2)-1}
\end{align}
\normalsize
We focus first on the left hand side (LHS) of Equation \eqref{eq:Start_R}:

\begin{align*}\label{eq:LS}
LHS=&\left(C_1(k_1,k_2)-u_s\right)\left(e^{u_s}\left[\left(k_1+\frac{k_2}{1+t}\right)+E(k_1,k_2)e^{-C_1(k_1,k_2)-1}\right]+a_0\right)
\end{align*}

For ease of notation, define $C_2(k_1,k_2)$ as:
\begin{equation}\label{eq:C2}
C_2(k_1,k_2)=\left(k_1+\frac{k_2}{1+t}\right)+E(k_1,k_2)e^{-C_1(k_1,k_2)-1}
\end{equation}

Rewriting the LHS using the value for $C_2(k_1,k_2)$:
\small
\begin{equation}\label{eq:LS_def}
LHS=\left(C_1(k_1,k_2)-u_s\right)\left[e^{u_s}\cdot C_2(k_1,k_2)+a_0)\right]   
\end{equation}
\normalsize

We now focus on the right side (RHS) of equation \eqref{eq:Start_R}:

\begin{align}\label{eq:RS_def}
RHS=&e^{u_s}\sum\limits_{i\in I_1}(u_i-u_s) + \frac{e^{u_s}}{1+t}\cdot\sum\limits_{i\in I_2}(u_i-u_s+\ln(1+t))\nonumber\\
&+E(k_1,k_2)\cdot \left(C_1(k_1,k_2)-u_s-1\right)e^{u_s-C_1(k_1,k_2)-1}\nonumber\\
RHS=&e^{u_s}\left[\sum\limits_{i\in I_1}u_i+\frac{1}{1+t}\cdot \sum\limits_{i\in I_2}u_i+\frac{k_2\ln(1+t)}{1+t}-u_s\left(k_1+\frac{k_2}{1+t}\right)\right] \nonumber\\
& +e^{u_s}e^{-C_1(k_1,k_2)-1}E(k_1,k_2)\cdot \left(C_1(k_1,k_2)-u_s+1\right)\nonumber\\
RHS=&e^{u_s}\cdot\left(k_1+\frac{k_2}{1+t}\right)\left[C_1(k_1,k_2)-u_s+1\right] +e^{u_s}e^{-C_1(k_1,k_2)-1}E(k_1,k_2)\cdot \left(C_1(k_1,k_2)-u_s+1\right)\nonumber\\
RHS=&e^{u_s}\cdot\left(C_1(k_1,k_2)-u_s+1\right)\cdot\underbrace{\left[\left(k_1+\frac{k_2}{1+t}\right)+E(k_1,k_2)\cdot e^{-C_1(k_1,k_2)-1}\right]}_{C_2(k_1,k_2)}\nonumber\\
RHS=&e^{u_s}\cdot\left(C_1(k_1,k_2)-u_s+1\right)\cdot C_2(k_1,k_2)
\end{align}

Putting together equations \eqref{eq:LS_def} and \eqref{eq:RS_def}, we have:
\small
\begin{align}\label{eq:us_rel}
LHS=&RHS\nonumber\\
\left(C_1(k_1,k_2)-u_s\right)\left[e^{u_s}\cdot C_2(k_1,k_2) +a_0)\right]=&e^{u_s}\cdot\left(C_1(k_1,k_2)-u_s+1\right)\cdot C_2(k_1,k_2)\nonumber\\
\left(C_1(k_1,k_2)-u_s\right)e^{u_s}\cdot C_2(k_1,k_2) +\left(C_1(k_1,k_2)-u_s\right)\cdot a_0=&\left(C_1(k_1,k_2)-u_s\right)e^{u_s}\cdot C_2(k_1,k_2) +e^{u_s}\cdot C_2(k_1,k_2)\nonumber\\
\left(C_1(k_1,k_2)-u_s\right)\cdot a_0=&e^{u_s}\cdot C_2(k_1,k_2)\nonumber\\
e^{u_s}=&-\frac{a_0}{C_2(k_1,k_2)}\cdot\left(u_s-C_1(k_1,k_2)\right)
\end{align}
\normalsize

Equation \eqref{eq:us_rel} has a known explicit closed form solution, and can be found using the following Lemma:

\begin{lemma}\label{lemma:W}
	Let $a,b \neq 0$ and $c$ be real numbers and $W(\cdot)$ be the Lambert function \citep{Corless1996}. The solution to the transcendental algebraic equation for $x$:
	\begin{equation}\label{eq:W_eq}
	e^{-ax}=b\left(x-c\right),
	\end{equation}
	is:
	\begin{equation}\label{eq:sol_W}
	x=c+\frac{1}{a}\cdot W\left(\frac{ae^{-ac}}{b}\right).
	\end{equation}
\end{lemma}
\noindent\textit{Proof of Lemma~\ref{lemma:W}.} Let us start with equation \eqref{eq:W_eq} and find an explicit solution to it.
\begin{align*}\label{proof_W}
e^{-ax}&=b\left(x-c\right)\\
e^{-ax + ac-ac}&=b\left(x-c\right) & \text{/multiplying both sides by $\frac{a}{b}\cdot e^{a(x-c)}$}q\\
\frac{a\cdot e^{-ac}}{b}& =a\cdot\left(x-c\right)\cdot e^{a(x-c)}& \text{/using definition of $W(\cdot)$ as in Eq. \eqref{ec:Lambert}}\\
W\left(\frac{ae^{-ac}}{b}\right)&= a\cdot(x-c)& \text{/reorganising and isolating $x$}\\
x&=c+\frac{1}{a}\cdot W\left(\frac{ae^{-ac}}{b}\right)& ,
\end{align*}
completing the proof. \qed

Identifying terms on equation \eqref{eq:us_rel},
the solution for $u_s$ is:

\begin{equation}\label{us_sol}
u_s=C_1(k_1,k_2)-W\left(\frac{C_2(k_1,k_2)}{a_0}\cdot e^{C_1(k_1,k_2)} \right)
\end{equation}

Let us call this value $u_s(k,k_1,k_2)$, meaning that is a function of the integers $k,k_1$ and $k_2$.
To get the revenue for this specific combination of parameters,
we can simply use equation \eqref{eq:R2_def}, giving us:

\begin{equation}\label{eq:Rk1k2}
R^{(k)}(k_1,k_2)=C_1(k_1,k_2)-u_s(k,k_1,k_2)
\end{equation}

Thus, the optimal revenue $\bm{R^{(k)}}$ given a specific integer $k$ can be obtained by:

\begin{equation}\label{eq:R_opt_JAP}
\bm{R^{(k)}} = \mathop{\max_{k_1, k_2 \geq 1}}_{k_1 + k_2 \leq k} R^{(k)}(k_1,k_2)
\end{equation}

Noting that there are $O(k^2)$ pairs $(k_1,k_2)$ to evaluate, the proof follows.\qed

\emph{Proof of Lemma~\ref{lem:sol_structure}}
The optimal revenue is already calculated in
Equation \eqref{eq:R_opt_JAP}.
The proof follows by first obtaining $u_s^*(k)$ from Equation \eqref{eq:R2}.
Then, for products in $I_1$, the price can be obtained directly
since their net utility is the same as $u_s^*(k)$. For products in $I_2$,
since $g_{1k}$ is satisfied with equality, all products share
the same net utility and equal to $u_s^*(k)-\ln(1+t)$.
Finally, for products in  $\bar{I}_k$, we can use the
relation provided in equation \eqref{KKT_p_o} to obtain the prices.
More explicitly, let $(k_1^*,k_2^*)$ be the integers satisfying  $\bm{R^{(k)}}=R^{(k)}(k_1^*,k_2^*)$.
To obtain the optimal prices, let $u_s^*(k)=u_s(k,k_1^*,k_2^*)$. By Equation
\eqref{eq:Rk1k2} $u_s^*(k)$ can be written as:
\small
\begin{equation}\label{eq:u_s_opt}
u_s^*(k)=\frac{(1+t)\sum_{i\in I_1}u_i +\sum_{i\in I_2}u_i+k_2^*\ln(1+t)}{k_1^*(1+t)+k_2^*}-1-\bm{R^{(k)}}
\end{equation}
\normalsize

Therefore, the optimal prices are given by:

\begin{equation}
\label{eq:opt_prices_kk1k2}
\bm{p^{(k)}_i(k)}=\begin{cases}
u_i- u_s^*(k) & \text{if } i\in I_1, \\
u_i- u_s^*(k)+\ln(1+t) & \text{if } i\in I_2,\\
1+\bm{R^{(k)}}  & \text{if } i\in \bar{I}_k	
\end{cases}
\end{equation}
\qed

\begin{lemma}\label{ex:fixed_price_bad}
	Fixed-Price policy can be arbitrarily bad for the Joint Assortment and Pricing problem under the Threshold Luce model.
\end{lemma}
\emph{Proof of Lemma~\ref{ex:fixed_price_bad}.}
Consider $N+1$ products, with product one having $u>0$ utility and $a_0=1$.
For all the remaining $N$ products let their utility to be: $\alpha u$,
with $\alpha <1$ such that in presence of product one,
all the rest of the products are ignored for threshold $t$.
The optimal revenue if we consider a fixed price strategy is \citep{li2011pricing,wang2012}:
\[R'=W(\exp(u-1))\]
Because no matter what fixed price we select, the $N$ lower utility products
are completely ignored and the first product is the only one contributing to the revenue,
and this is the best revenue that we can achieve given that.
Now, let us consider the optimal revenue obtained with the strategy described in Theorem \eqref{thm:ThrehsoldOpt}

\begin{equation}\label{eq:rev_alg_example}
R^*=W\left(\left[\frac{(1+t)+N}{1+t}\right]\cdot \exp\left(\frac{(1+t)(u-1) +N(\alpha u-1)+N\ln(1+t)}{(1+t)+N}\right)\right),
\end{equation}

let us find an explicit relation between $R'$ and $R^*$. Starting from equation \eqref{eq:rev_alg_example}:

\begin{align*}
R^*=&W\left(\left[\frac{(1+t)+N}{1+t}\right]\cdot \exp\left(\frac{(1+t)(u-1) +N(\alpha u-1)+N\ln(1+t)}{(1+t)+N}\right)\right)\\
R^*=&W\left(\left[\frac{(1+t)+N}{1+t}\right]\cdot \exp\left((u-1)\cdot\frac{(1+t)+N\alpha}{(1+t)+N}+\frac{N\ln(1+t)}{(1+t)+N}\right)\right)\\
R^*=&W\left(\left[\frac{(1+t)+N}{1+t}\right]\cdot \exp\left((u-1)+\frac{N}{(1+t)+N}\cdot\underbrace{\left(\ln(1+t)-(u-1)\cdot(1-\alpha)\right)}_{\Gamma}\right)\right)\\	
\end{align*}

We know that the Lambert function is concave, increasing and unbounded \citep{Corless1996,li2011pricing}.
With this in mind, let $u$ be such that $\Gamma$ is greater or equal than zero (for example, setting $u=1.9$, $\alpha=0.5$ and $t=0.5$, makes $\Gamma>0$ and product $1$ dominates the rest of the products), this is:

\begin{equation}\label{eq:cond_u}
\frac{\ln(1+t)}{1-\alpha} +1 \geq u.
\end{equation}
Using this, we have:
\begin{equation}\label{eq:last_bad_fixed_price}
R^*\geq W\left(\left[\frac{(1+t)+N}{1+t}\right]\cdot \exp\left(u-1\right)\right)\\	
\end{equation}

Where the argument of the Lambert function is exactly the same as $R^*$, but multiplied by a constant factor larger than one and depending on $N$. Putting everything together, we have:

\begin{equation}\label{ineq:opt_rel}
R^*\geq  W\left(\left[\frac{(1+t)+N}{1+t}\right]\cdot \exp\left(u-1\right)\right)\geq R'
\end{equation}

The expression in the middle can be arbitrarily larger than $R'$
by letting $N$ tend to infinity, and so is $R^*$.
Thus, the fixed price policy can be arbitrarily bad under the Threshold Luce model.\qed

\newpage
\section{Numerical Experiments}\label{sec:numerical}
This section present numerical results on the performance of the algorithms developed in
Sections \ref{sec:assortment} and \ref{sec:JAP_TLM},
compared against classical algorithms in the literature
such as revenue-ordered assortments for the assortment problem,
and pricing policies like Fixed-Price \citep{li2011pricing} which is optimal
for the conventional MNL, and Quasi-Same price \citep{wang2017impact}, which is optimal for the
proposed variant of the MNL including search cost. Quasi-Same Price amounts to
have a fixed price for all products but one (the one having the smallest utility).

We also provide some insights on the factors
that influence the nature of the solution and provide some explanation
on the difference in performance between the different strategies.


\subsection{Assortment Optimisation}\label{subsec:numerical_assortment}

This section presents some numerical results on the performance of revenue-ordered assortments
(RO) against our proposed strategy detailed in Section \ref{sec:assortment}, which we call
\texttt{2SLM-OPT}. In order to do this, we variate the number of products $n$,
the attractiveness of the outside option $a_0$ and the density $d$ of the graph,
which we use as the probability that a dominance relation is active
for each pair of products\footnote{used as the probability that an edge in the dominance graph occurs}.
Theoretically, as shown in Example  \ref{ex:ro_not_optimal}, the optimality gap can be as large desired.
But 	in practice, we were able to found gaps as large as $95.40\%$.

Each tested family or class of instances is defined by essentially three numbers:
the number of products $n$, the attractiveness of the outside option $a_0$,
and the density $d$, that controls the probability that
a dominance edge exists, and then we also compute the transitive closure
over the resulting graph. It is worth noticing that we did not consider
the case $a_0=0$ because in those cases, the optimal solution is simply selecting the highest
revenue product and therefore both strategies coincide. In total, we experimented
with $48$ classes or families of instances, each containing $250$ instances.
In each specific instance, revenues and
utilities are drawn from an uniform distribution between $0$ and $10$.
We ran 	both strategies (RO and \texttt{2SLM-OPT}) and report the
average and worst optimality gap for the RO
strategy. We are not providing running times, because as expected,
\texttt{2SLM-OPT} takes more time than RO,  but all instances
can be solved very fast in practice (less than half a second).
Table \ref{tab:assortment_simple_subtotal} presents the results which can be summarized as follows:

\begin{enumerate}
	
	\item The average gap tends to increase with the number of products,
	reaching about $14\%$ for $30$ products.
	The worst gap is more instance-dependent
	(as it strongly depends on the dominance structure,
	and how revenues are matched with attractiveness) so it can be
	large both in smaller and larger instances. However,
	it tends to increase with the density of the dominance graph,
	as it is more likely for RO to choose
	a product that dominates potential contributors whose inclusion
	can be more profitable than keeping the higher attractiveness one.	
	
	\item The average gap generally widens as the outside option
	attractiveness increases. With a high outside
	option, we typically expect to select more products to counterbalance
	the effect of the no-choice alternative.
	This can amplify the difference between \texttt{2SLM-OPT} and RO
	as the likelihood that	the optimal solution turns out to be
	revenue-ordered decreases, given the randomness of the dominance relation.
	
	\item With higher densities, is more likely to \textit{make a mistake} and include a product that dominates
	many \textit{potential contributors} that considered together, might be more profitable. Thus,
	both the average and worst gap widens as the density increases in general. The exception occurs at the higher end of
	densities where not many products can be included without provoking dominances. Here the solutions
	of both strategies tends to be similar and select a few higher revenue products.
	This is also interesting from a managerial standpoint: when customers have more clarity
	on what products are clearly superior in comparison, this might drift the offered assortment to be smaller,
	compared against when customers does not have a clear hierarchy among products.
	
\end{enumerate}	
	
\begin{table}[!htbp]
	\centering
	\begin{adjustbox}{totalheight=\textheight-2.5\baselineskip}
		\begin{tabular}{lrrrlr}
			\toprule
			\multicolumn{1}{c}{\multirow{2}[4]{*}{$(n,a_0,d)$}} & \multicolumn{3}{c}{RO Assortments} &       & \multicolumn{1}{c}{\texttt{2SLM-OPT}} \\
			\cmidrule{2-4}\cmidrule{6-6}          & \multicolumn{1}{l}{Avg. Gap (\%)} & \multicolumn{1}{l}{Worst Gap (\%)} & \multicolumn{1}{l}{Avg. Cardinality} &       & \multicolumn{1}{l}{Avg. Cardinality} \\
			\midrule
			\midrule
			(5,1,0.2) & 0.476 & 48.899 & 1.484 &       & 1.496 \\
			(5,1,0.4) & 1.532 & 80.404 & 1.384 &       & 1.376 \\
			(5,1,0.8) & 2.812 & 71.888 & 1.08  &       & 1.084 \\
			(5,2,0.2) & 1.173 & 73.387 & 1.82  &       & 1.816 \\
			(5,2,0.4) & 1.827 & 69.8  & 1.504 &       & 1.536 \\
			(5,2,0.8) & 4.529 & 94.759 & 1.116 &       & 1.14 \\
			(5,4,0.2) & 1.835 & 69.574 & 2.108 &       & 2.104 \\
			(5,4,0.4) & 3.133 & 61.627 & 1.784 &       & 1.82 \\
			(5,4,0.8) & 5.378 & 69.555 & 1.2   &       & 1.228 \\
			(5,8,0.2) & 1.789 & 64.546 & 2.284 &       & 2.34 \\
			(5,8,0.4) & 5.927 & 70.854 & 1.884 &       & 1.988 \\
			(5,8,0.8) & 6.933 & 91.335 & 1.168 &       & 1.244 \\
			\midrule
			Avg. $n=5$ & 3.112 & 72.219 & 1.568 &       & 1.597667 \\
			&       &       &       &       &  \\
			(10,1,0.2) & 0.68  & 51.339 & 1.872 &       & 1.896 \\
			(10,1,0.4) & 2.388 & 63.414 & 1.5   &       & 1.524 \\
			(10,1,0.8) & 3.997 & 95.49 & 1.092 &       & 1.076 \\
			(10,2,0.2) & 1.385 & 49.292 & 2.272 &       & 2.296 \\
			(10,2,0.4) & 3.275 & 90.659 & 1.604 &       & 1.664 \\
			(10,2,0.8) & 6.495 & 73.787 & 1.132 &       & 1.148 \\
			(10,4,0.2) & 1.984 & 61.872 & 2.612 &       & 2.764 \\
			(10,4,0.4) & 5.734 & 90.983 & 1.92  &       & 2.112 \\
			(10,4,0.8) & 7.107 & 86.55 & 1.156 &       & 1.224 \\
			(10,8,0.2) & 3.509 & 41.995 & 3.08  &       & 3.2 \\
			(10,8,0.4) & 6.592 & 82.358 & 2.028 &       & 2.304 \\
			(10,8,0.8) & 8.916 & 92.576 & 1.172 &       & 1.304 \\
			\midrule
			Avg. $n=10$ & 4.3385 & 73.35958333 & 1.786666667 &       & 1.876 \\
			&       &       &       &       &  \\
			(20,1,0.2) & 1.067 & 36.45 & 2.18  &       & 2.216 \\
			(20,1,0.4) & 2.664 & 82.68 & 1.448 &       & 1.556 \\
			(20,1,0.8) & 2.884 & 74.534 & 1.1   &       & 1.092 \\
			(20,2,0.2) & 2.349 & 40.095 & 2.46  &       & 2.652 \\
			(20,2,0.4) & 3.452 & 41.717 & 1.696 &       & 1.856 \\
			(20,2,0.8) & 5.112 & 83.79 & 1.132 &       & 1.192 \\
			(20,4,0.2) & 3.786 & 34.659 & 2.84  &       & 3.184 \\
			(20,4,0.4) & 8.575 & 73.075 & 1.848 &       & 2.14 \\
			(20,4,0.8) & 7.749 & 86.321 & 1.152 &       & 1.284 \\
			(20,8,0.2) & 5.938 & 68.465 & 3.352 &       & 3.856 \\
			(20,8,0.4) & 8.88  & 52.627 & 2.088 &       & 2.616 \\
			(20,8,0.8) & 10.204 & 94.021 & 1.152 &       & 1.392 \\
			\midrule
			Avg. $n=20$ & 5.221666667 & 64.03616667 & 1.870666667 &       & 2.086333 \\
			&       &       &       &       &  \\
			(30,1,0.2) & 1.762 & 20.877 & 2.068 &       & 2.228 \\
			(30,1,0.4) & 3.34  & 83.702 & 1.44  &       & 1.616 \\
			(30,1,0.8) & 3.773 & 62.764 & 1.056 &       & 1.108 \\
			(30,2,0.2) & 3.084 & 43.736 & 2.544 &       & 2.864 \\
			(30,2,0.4) & 5.554 & 79.378 & 1.64  &       & 1.968 \\
			(30,2,0.8) & 5.499 & 86.544 & 1.072 &       & 1.148 \\
			(30,4,0.2) & 4.721 & 53.873 & 2.984 &       & 3.464 \\
			(30,4,0.4) & 8.046 & 74.267 & 1.876 &       & 2.3 \\
			(30,4,0.8) & 9.045 & 92.51 & 1.14  &       & 1.304 \\
			(30,8,0.2) & 7.623 & 46.498 & 3.368 &       & 4.188 \\
			(30,8,0.4) & 14.266 & 91.851 & 1.916 &       & 2.684 \\
			(30,8,0.8) & 11.422 & 75.239 & 1.132 &       & 1.412 \\
			\midrule
			Avg. $n=30$ & 6.51125 & 67.60325 & 1.853 &       & 2.190333 \\
			\bottomrule
		\end{tabular}%
	\end{adjustbox}
	\caption{Numerical experiments comparing the revenue ordered assortment strategy (RO)
		and our proposed strategy \texttt{2SLM-OPT}.
		For each class of instances, we display the average optimality gap
		and the worst-case gap, as well as the computing time and the cardinality of the offered set.}
	\label{tab:assortment_simple_subtotal}%
\end{table}%

\newpage 
\subsection{Joint Assortment and Pricing Optimisation}\label{subsec:numerical_JAP}

This section presents some numerical results related to solve the Joint Assortment and Pricing Problem
discussed in Section \ref{sec:JAP_TLM}. We analyse the performance of algorithm \texttt{TLM-Opt},
compared against Fixed-Price strategy, which is optimal for the MNL and Quasi-Same Price strategy \citep{wang2017impact},
which is optimal for the MNL variant considered in their paper that takes into consideration search cost,
and it basically a fixed price for all products but one, which share some similarities with our
proposed pricing policy, as it is fixed price in general but the higher and lower ends of the utility spectrum.

Each tested family or class of instances is characterized by three numbers:
the number of products $n$; the threshold $t$,
that controls how tolerant are customers with respect
to differences in attractiveness and the attractiveness of the outside option $a_0$, which controls
how likely is that customers 	review all products without purchasing.
In total, we experimented with $48$ classes or families of instances,
each containing $250$ instances. In each specific instance, revenues and
utilities are drawn from an uniform distribution between $0$ and $10$. We ran
the three strategies: Fixed Price, Quasi-Same Price and \texttt{TLM-Opt},
and report the average and worst optimality gap for Fixed Price and Quasi-Same Price
strategies, as well as cardinality of the offered set for both strategies.
These numerical experiments were conducted in Python 3.6 on a computer
with 8 processors (each with 3.6 GHz CPU) and $16$ GB of RAM.
Table \ref{tab:pricing_simple_subtotal} presents the results which can be summarized as follows:

\begin{enumerate}
	
	\item As expected \texttt{TLM-Opt} outperformed the other two algorithms in terms of revenue,
	and being quite fast to execute (less than half of a second for all the instances simulated).
	\item Fixed-Price policy performs the worst across the board,
	which is expected given that it has the lowest degrees of freedom,
	as shown in example \ref{exmp:fixed_price_num}. Although the average
	gap is quite low, it can be as high as $43.027\%$. In fact,
	fixed-price policy can be arbitrarily bad. A proof of this fact
	is provided in Appendix \ref{App:proofs}, Lemma \ref{ex:fixed_price_bad}.
	\item Quasi-Same price policy also performs well on average,
	and the worst gap obtained was $29.964\%$, which is significantly better than the worst
	gap for Fixed Price policy.
	\item The cardinality of the optimal solution is always at least the same or greater than
	Fixed-Price policy. This can be observed empirically, or deduced
	analytically. The intuition behind it is that given the functional form of the revenue
	for Fixed-Price and the fact that the Lambert function is strictly increasing
	the strategy always try to show as much as possible.
	This, and the fact that under same price,
	the dominance relation only depends upon intrinsic utilities,
	imply that there is a limit  on the number of products
	that the fixed price policy can offer
	\footnote{the last product `k' where $\frac{a_1(p_1)}{a_k(p_k)}=\exp(u_1 – u_k) \leq (1+t)$}
	without causing any domination for  low intrinsic utility products.
	On the other hand, under \texttt{TLM-Opt} (or Quasi Same price)
	we can go further and add products
	in such a way that the dominance relations are not triggered,
	and therefore we can include more products.   
	
	\item The main difference stems from the fact that our strategy
	leverage both ends of the utility spectrum,
	and reveals the following interesting insight.
	Sometimes in order to avoid low attractiveness products to be dominated,
	we want to: increase the price of the higher utility products
	(to make them less attractive) and at the same time,
	reduce the price for lower utility products,
	in order to make them more attractive, and making them
	\textit{visible} for the consumer.
	
\end{enumerate}

\begin{table}[!htbp]
	\centering
	\resizebox{.85\textwidth}{!}{\begin{tabular}{lrrrlrrrlr}
			\toprule
			\multicolumn{1}{c}{\multirow{2}[4]{*}{$(n,t,a_0)$}} & \multicolumn{3}{c}{Fixed Price} &       & \multicolumn{3}{c}{Quasi Same Price} &       & \multicolumn{1}{c}{\texttt{TLM-Opt}} \\
			\cmidrule{2-4}\cmidrule{6-8}\cmidrule{10-10}          & \multicolumn{1}{l}{Avg. Gap (\%)} & \multicolumn{1}{l}{Worst Gap (\%)} & \multicolumn{1}{l}{Avg. Cardinality} &       & \multicolumn{1}{l}{Avg. Gap (\%)} & \multicolumn{1}{l}{Worst Gap (\%)} & \multicolumn{1}{l}{Avg. Cardinality} &       & \multicolumn{1}{l}{Avg. Cardinality} \\
			\midrule
			\midrule
			(5,0.5,1) & 2.164728 & 17.514 & 1.212 &       & 0.442364 & 8.609 & 2.212 &       & 1.92 \\
			(5,0.5,10) & 2.925244 & 27.779 & 1.26  &       & 0.54798 & 11.074 & 2.248 &       & 1.888 \\
			(5,0.5,100) & 4.136064 & 43.027 & 1.24  &       & 1.16004 & 29.964 & 2.108 &       & 1.856 \\
			(5,1,1) & 1.575348 & 13.446 & 1.384 &       & 0.253996 & 4.638 & 2.38  &       & 2.028 \\
			(5,1,10) & 2.074672 & 24.984 & 1.472 &       & 0.362416 & 11.116 & 2.448 &       & 2.04 \\
			(5,1,100) & 2.726188 & 33.938 & 1.416 &       & 0.52712 & 14.404 & 2.308 &       & 1.952 \\
			(5,2,1) & 0.9865 & 8.881 & 1.58  &       & 0.723548 & 8.881 & 1.812 &       & 2.132 \\
			(5,2,10) & 1.4685 & 10.592 & 1.536 &       & 1.040004 & 10.592 & 1.84  &       & 2.116 \\
			(5,2,100) & 2.343244 & 32.77 & 1.624 &       & 1.520068 & 22.581 & 1.924 &       & 2.196 \\
			(5,5,1) & 0.415064 & 5.103 & 2.012 &       & 0.153776 & 3.748 & 2.64  &       & 2.468 \\
			(5,5,10) & 0.918528 & 17.044 & 1.844 &       & 0.460556 & 17.044 & 2.424 &       & 2.384 \\
			(5,5,100) & 1.092544 & 11.539 & 1.972 &       & 0.466548 & 7.574 & 2.552 &       & 2.468 \\
			\midrule
			Avg. $n=5$ & 1.902219 & 20.55142 & 1.546 &       & 0.638201 & 12.51875 & 2.241333 &       & 2.120667 \\
			&       &       &       &       &       &       &       &       &  \\
			(10,0.5,1) & 3.63332 & 14.951 & 1.408 &       & 1.079656 & 6.774 & 2.408 &       & 2.896 \\
			(10,0.5,10) & 4.710012 & 30.328 & 1.512 &       & 1.499744 & 15.575 & 2.512 &       & 3.132 \\
			(10,0.5,100) & 6.7165 & 24.489 & 1.42  &       & 2.465028 & 23.517 & 2.364 &       & 2.912 \\
			(10,1,1) & 2.69928 & 12.027 & 1.748 &       & 0.835872 & 7.352 & 2.748 &       & 3.264 \\
			(10,1,10) & 3.563196 & 15.769 & 1.672 &       & 1.264116 & 10.296 & 2.672 &       & 3.18 \\
			(10,1,100) & 4.822544 & 26.928 & 1.756 &       & 1.704816 & 15.125 & 2.74  &       & 3.264 \\
			(10,2,1) & 1.38662 & 8.541 & 2.076 &       & 0.803492 & 8.541 & 2.576 &       & 3.236 \\
			(10,2,10) & 2.37252 & 15.852 & 2.016 &       & 1.284344 & 15.852 & 2.544 &       & 3.38 \\
			(10,2,100) & 3.115392 & 18.694 & 2.076 &       & 1.53734 & 18.694 & 2.612 &       & 3.34 \\
			(10,5,1) & 0.611308 & 4.19  & 2.888 &       & 0.322156 & 2.961 & 3.492 &       & 3.908 \\
			(10,5,10) & 0.931108 & 5.537 & 2.804 &       & 0.523108 & 4.975 & 3.432 &       & 3.84 \\
			(10,5,100) & 1.323312 & 12.683 & 2.828 &       & 0.705452 & 12.683 & 3.44  &       & 3.936 \\
			\midrule
			Avg. $n=10$ & 2.990426 & 15.83242 & 2.017 &       & 1.16876 & 11.86208 & 2.795 &       & 3.357333 \\
			&       &       &       &       &       &       &       &       &  \\
			(20,0.5,1) & 5.227892 & 16.189 & 1.964 &       & 2.406408 & 10.383 & 2.964 &       & 5.412 \\
			(20,0.5,10) & 6.505472 & 18.556 & 1.844 &       & 2.926688 & 11.734 & 2.844 &       & 4.868 \\
			(20,0.5,100) & 9.65628 & 30.904 & 1.844 &       & 4.633812 & 20.602 & 2.84  &       & 5.104 \\
			(20,1,1) & 3.917928 & 11.225 & 2.332 &       & 1.87528 & 7.241 & 3.332 &       & 5.476 \\
			(20,1,10) & 4.640684 & 20.635 & 2.32  &       & 2.227456 & 14.112 & 3.32  &       & 5.384 \\
			(20,1,100) & 6.765772 & 26.431 & 2.368 &       & 3.075284 & 17.929 & 3.368 &       & 5.48 \\
			(20,2,1) & 2.197276 & 9.372 & 3.324 &       & 1.210484 & 9.372 & 4.112 &       & 6.164 \\
			(20,2,10) & 2.669532 & 10.396 & 3.316 &       & 1.449576 & 9.11  & 4.168 &       & 6.252 \\
			(20,2,100) & 3.708316 & 15.228 & 3.28  &       & 2.055808 & 14.009 & 4.092 &       & 5.924 \\
			(20,5,1) & 0.878752 & 4.584 & 4.636 &       & 0.577612 & 4.584 & 5.236 &       & 6.976 \\
			(20,5,10) & 1.244528 & 5.209 & 4.632 &       & 0.805216 & 5.209 & 5.26  &       & 7.1 \\
			(20,5,100) & 1.718416 & 11.142 & 4.664 &       & 1.138056 & 7.856 & 5.312 &       & 7.192 \\
			\midrule
			Avg. $n=20$ & 4.094237 & 14.98925 & 3.043667 &       & 2.031807 & 11.01175 & 3.904 &       & 5.944333 \\
			&       &       &       &       &       &       &       &       &  \\
			(30,0.5,1) & 6.343964 & 16.145 & 2.24  &       & 3.642808 & 10.606 & 3.24  &       & 7.344 \\
			(30,0.5,10) & 8.202948 & 22.238 & 2.224 &       & 4.572832 & 15.339 & 3.224 &       & 7.32 \\
			(30,0.5,100) & 10.6693 & 26.659 & 2.228 &       & 6.018692 & 18.973 & 3.224 &       & 7.2 \\
			(30,1,1) & 4.0957 & 14.803 & 3.044 &       & 2.266944 & 10.581 & 4.044 &       & 7.464 \\
			(30,1,10) & 5.039272 & 19.686 & 3.24  &       & 2.906664 & 13.583 & 4.24  &       & 7.892 \\
			(30,1,100) & 7.451636 & 23.606 & 3.084 &       & 4.349844 & 14.428 & 4.084 &       & 7.884 \\
			(30,2,1) & 2.193896 & 12.267 & 4.344 &       & 1.32162 & 9.526 & 5.252 &       & 8.588 \\
			(30,2,10) & 3.110752 & 13.315 & 4.252 &       & 1.911924 & 9.177 & 5.148 &       & 8.64 \\
			(30,2,100) & 4.005252 & 19.963 & 4.476 &       & 2.473688 & 19.963 & 5.4   &       & 8.684 \\
			(30,5,1) & 1.023416 & 4.25  & 6.26  &       & 0.74478 & 4.25  & 6.88  &       & 10.132 \\
			(30,5,10) & 1.328936 & 5.534 & 6.328 &       & 0.93428 & 4.053 & 7.012 &       & 10.232 \\
			(30,5,100) & 1.730272 & 6.284 & 6.436 &       & 1.189296 & 6.164 & 7.136 &       & 10.104 \\
			\midrule
			Avg. $n=30$ & 4.599612 & 15.39583 & 4.013 &       & 2.694448 & 11.38692 & 4.907 &       & 8.457 \\
			\bottomrule
	\end{tabular}}%
	\caption{Numerical experiments comparing Fixed-Price and Quasi-Same price against \texttt{TLM-Opt}.
		For each class of instances, for non-optimal strategies we display the average optimality gap, worst-case gap  and the cardinality of the offered set. We also provide the average of those metrics for each value of $n$ considered.}
	\label{tab:pricing_simple_subtotal}%
\end{table}%
	%
	%
	%
	%

\end{document}